\documentclass[a4paper,USenglish,cleveref, autoref, thm-restate]{socg-lipics-v2021}

\pdfoutput=1 

\graphicspath{{./graphics/}}

\usepackage{amsmath}
\usepackage{bm}
\usepackage{cases}
\usepackage{nicefrac}
\usepackage{multirow}
\usepackage{lineno}

\usepackage{tikz}
\usetikzlibrary{positioning}
\usetikzlibrary{calc}
\usetikzlibrary {arrows.meta}

\newcommand{\enne}{\mathbb{N}}
\newcommand{\erre}{\mathbb{R}}
\newcommand{\ip}{\text{IP}}
\newcommand{\lp}{\text{LP}}
\newcommand{\cP}{\mathcal{P}}
\newcommand{\cC}{\mathcal{C}}
\newcommand{\cM}{\mathcal{M}}

\newcommand{\ig}{\text{IG}}
\newcommand{\clp}{\text{CLP}}

\title{Integrality gap preserving reductions}

\titlerunning{IG preserving reductions} 

\author{{Koppány István} Encz\footnote{Corresponding author}}{Faculty of Informatics, Università della Svizzera italiana, [CH 6962 Lugano, Switzerland] \and Istituto Dalle Molle di studi sull'intelligenza artificiale (IDSIA USI-SUPSI), [CH 6962 Lugano, Switzerland] \and \url{https://sites.google.com/view/enczkoppany/home}}{enczk@usi.ch}{https://orcid.org/0009-0009-0937-6809}{}

\author{Monaldo Mastrolilli}{Dipartimento Tecnologie innovative, Scuola universitaria professionale della Svizzera italiana [CH 6962 Lugano, Switzerland] \and Istituto Dalle Molle di studi sull'intelligenza artificiale (IDSIA USI-SUPSI) [CH 6962 Lugano, Switzerland]}{monaldo.mastrolilli@supsi.ch}{https://orcid.org/0000-0002-2948-9749}{}

\author{Eleonora Vercesi}{Faculty of Informatics, Università della Svizzera italiana, [CH 6962 Lugano, Switzerland] \and Istituto Dalle Molle di studi sull'intelligenza artificiale (IDSIA USI-SUPSI), [CH 6962 Lugano, Switzerland] \and \url{https://eleonoravercesi.github.io/}}{eleonora.vercesi@usi.ch}{https://orcid.org/0000-0002-1621-2484}{}

\authorrunning{K. I. Encz and M. Mastrolilli and E. Vercesi}

\Copyright{Jane Open Access and Joan R. Public} 



\relatedversion{} 

\nolinenumbers 

\EventEditors{John Q. Open and Joan R. Access}
\EventNoEds{2}
\EventLongTitle{42nd Conference on Very Important Topics (CVIT 2016)}
\EventShortTitle{CVIT 2016}
\EventAcronym{CVIT}
\EventYear{2016}
\EventDate{December 24--27, 2016}
\EventLocation{Little Whinging, United Kingdom}
\EventLogo{}
\SeriesVolume{42}
\ArticleNo{23}

\begin{document}

\maketitle

\nolinenumbers

\begin{abstract}
We propose a framework for the systematic study of integrality gaps of combinatorial optimization problems with respect to a fixed linear programming formulation. The method, called \emph{integrality gap preserving reduction}, consists of iteratively shrinking the input universe of the problem while guaranteeing that gap-maximizing instances remain selected. When the subset of remaining instances becomes specific enough, we calculate the integrality gap explicitly. 

Besides applying integrality gap preserving reductions to three well-known optimization problems via their standard linear programming formulations (weighted vertex cover problem, multiple knapsack problem, and unrelated machine scheduling problem), we analyse the restricted assignment problem via its configuration LP relaxation. We prove that the integrality gap is equal to $1$ for three ``easy'' subclasses of the problem that are either solvable in polynomial time or admit a PTAS (e.g., the all-one processing time case). For some remaining cases, we improve the current lower bound using our technique.
\end{abstract}

\section{Introduction}

Many combinatorial optimization problems can be formulated with an \emph{integer linear program} (IP)
\begin{equation}\label{ip}
  \ip(I) := \min\{cx: Ax \ge b, x \in \enne^n\},
\end{equation}
where $I=(A,b,c)$ is an instance of a minimization problem $\cP$, $A \in \mathbb{Q}^{m \times n}$, $b \in \mathbb{Q}^m$, and $c \in \mathbb{Q}^n$. Since solving \eqref{ip} is often $\mathsf{NP}$-hard, one typically considers its \emph{linear relaxation}
\begin{equation}\label{lp}
    \lp(I):= \min\{cx: Ax \ge b, \, x \in \erre^n\},
\end{equation}
obtained by dropping integrality. Clearly, $\lp(I) \le \ip(I)$, and $\lp(I)$ can be solved in polynomial time \cite{karmarkar, khachiyan}.

A key issue to bear in mind is the loss of solution quality. A straightforward measure of this loss is the \emph{integrality gap} \footnote{For maximization problems, $\ig(I)=\lp(I)/\ip(I)$.} of an instance
\begin{equation}\label{gap}
    \ig(I):= \frac{\ip(I)}{\lp(I)}. 
\end{equation}
From a computational complexity perspective, determining $\ig(I)$ is as hard as solving $\ip(I)$, so most efforts focus on bounding it instead. For an entire class $\cP$ of optimization problem instances given as a collection $S_0$, define the global gap as
\begin{equation}\label{Gap}
    \ig(\cP) := \sup\{\ig(I): I \in S_0\}.
\end{equation}
Unlike the single-instance case, $\ig(\cP)$ is often tractable and known for several problems such as unrelated machine scheduling or multi-knapsack \cite{machine_scheduling_review, martello_toth}.
Obtaining \eqref{Gap} typically involves two steps. A lower bound is given by a family $\{I_k\}$ yielding $l := \sup_k \ig(I_k) \le \ig(\cP)$. An upper bound is usually extracted from approximation algorithms: an algorithm $M$ is an \emph{$\alpha$-approximation} (with respect to a fixed lower bound $\lp$) if $c(M(I)) \le \alpha \cdot \lp(I)$ for all $I \in \cP$, implying $\ig(\cP)\le \alpha$. The goal is to minimize $\alpha - l$, ideally achieving equality.

This approach succeeds in several cases. For unrelated machine scheduling with $m$ machines, a simple instance $I$ with only $1$ job and $p_{1,1}= \ldots = p_{1,m}=1$ gives $\ig(I)=m$, and an elementary rounding procedure serves as an $m$-approximation algorithm, proving $\ig(\cP)=m$. In contrast, for other problems, decades of effort have been insufficient for closing the distance between $l$ and $\alpha$. A prominent example is the integrality gap of the Dantzig - Fulkerson - Johnson formulation for the Symmetric Travelling Salesperson Problem (STSP), conjectured to have a gap $\frac{4}{3}$ \cite{Williamson90}. For STSP, the lower and upper bounds remain $\frac{4}{3}$ \cite{benoit08} and $\frac{3}{2}$ \cite{christofides1976report} (slightly improved to $\frac{3}{2}-10^{-36}$ in \cite{karlin2024improved}), leaving a persistent gap.

Beyond approximation algorithms, few alternative methods have been studied. Most results are problem-specific: Singh \cite{singh19} derives an explicit formula for vertex cover, while Villa and et al. \cite{villa2025integrality} prove the gap is $\frac{4}{3}$ for a subclass of STSP using computer-assisted methods building on Benoit and Boyd \cite{benoit08}. A sporadic general-purpose method appears in \cite{carr23}.

\bigskip

We propose a framework focusing solely on the integrality gap without relying on external factors, and aiming to determine $\ig(\mathcal{P})$ directly. The method stems from a simple empirical observation: both known (e.g., unrelated machine scheduling, knapsack) and conjectured (STSP, restricted assignment problem \cite{Jansen17}) worst-case instances in terms of integrality gap exhibit an extremely symmetric structure; let us recall the instance $p_{1,1}=\ldots = p_{1,m}=1$ from before. Apart from the problems we discuss in this paper, we refer to some highly symmetric families of STSP instances \cite{benoit08,hougardy2014integrality,hougardy2021hard,zhong2025lower}, achieving the best known asymptotic lower bounds on IG. The observation is nicely complemented by experimental results in \cite{Vercesi2023IGTSP} (table 5) which claim that randomly sampled STSP instances have a low integrality gap with high probability. 

These general phenomena motivate the study of \emph{integrality gap preserving reductions} (\emph{gap-preserving reductions} or \emph{gap reductions} for short): given a problem instance, modify it locally so that $1)$ the integrality gap does not decrease, and $2)$ the modified instance exhibits more structure with respect to the original one, according to some appropriately defined notion of ``structuredness''. We elaborate on this notion in Section \ref{sec:gap_red}. The methodology has two notable advantages. First, from a structural and pedagogical standpoint, rather than relying on problem-specific ad-hoc arguments often needed to bound integrality gaps, our framework offers a conceptually unified perspective. Second, it highlights the underlying relationship between the symmetry of an instance and its integrality gap: not only do we obtain the single numeric value $\ip(\cP)$, but also a hierarchy of increasingly symmetric instances, in which moving upwards both increases the gap and strengthens the structure.

We organise the paper as follows: in Section \ref{sec:gap_red}, we define gap reductions and share our first observations. To show that our framework is useful and works for different cases, we use it to find the integrality gap for three classic optimization problems (via their natural LP-relaxations) for which the gap was already known: the weighted vertex cover, the multiple knapsack and the unrelated machine scheduling problems. In Sections \ref{sec:vertex_cover}, \ref{sec:multi_knap} and \ref{sec:machine_sched}, we show that our method reaches the same known results by focusing on the underlying structure of the problems. This allows us to use a single, consistent approach instead of the different specific tricks usually required for each problem. Full technical details are provided in Appendix \ref{app:vertex_cover}, \ref{app:multi_knap} and \ref{app:machine_sched}. In Section \ref{sec:conifg_lp}, we advance the current knowledge regarding the still unknown integrality gap of the restricted assignment problem with the configuration LP relaxation by using gap-preserving reductions. In particular, we assert that the gap is $1$ for some easy (solvable in polynomial time or admitting a PTAS) sub-problems: the two-machine case in Theorem \ref{thm:p2} and the all-one processing time case in Theorem \ref{thm:all_one}. Lastly, in Section~\ref{sec:conclusion}, we formulate an intriguing hypothesis and state that, if it holds, the integrality gap of the configuration LP coincides with that of the graph balancing problem. Full details are presented in Appendix \ref{app:config}.

\section{Gap-preserving reductions}\label{sec:gap_red}

The baseline of our approach is the empirical observation that gap-maximizing instances tend to be highly structured. We can exploit it in the following manner: given an initial instance of an optimization problem $\cP$, iteratively alter it so that each new instance has a larger-or-equal gap than the last, and each new instance is ``simpler'' than the old one in terms of an intuitive metric. Frequently appearing local changes can be grouped into two major categories: structural changes to an instance (e.g., reducing the number of jobs/machines/knapsacks/items) which boost simplification, or numerical modifications of the processing times/item weights/item profits that try to make the instance more uniform. When the current instance is sufficiently structured, we calculate its gap explicitly and obtain $\ig(\mathcal{P})$ as the maximal gap among the final instances. We give precise definitions below.

Consider a minimization problem $\cP$ and its corresponding instance set $S_0$, where $I \in S_0$ is determined by a triple $(A,b,c): A \in \mathbb{Q}^{m \times n}, b \in \mathbb{Q}^m$ and $c \in \mathbb{Q}^n$ for some $n,m \in \enne$. For concreteness, when $\cP$ is the weighted vertex cover problem in its natural IP-formulation (defined in Section \ref{sec:vertex_cover}), $A = \begin{pmatrix} A_G \\I_n \end{pmatrix}$ where $A_G$ is the edge-node incidence matrix of a graph $G$ on $n$ nodes and $I_n$ is the $n\times n$ identity matrix, $b = \begin{pmatrix} 1_m \\ 0_n \end{pmatrix}$ where $1_m$ and $0_n$ are the length-$m$ all-one vector and the length-$n$ zero vector, and $c=w$. The optima of the integer and linear programs are given as \eqref{ip} and \eqref{lp}. The integrality gap of $I$ is given by \eqref{gap}. 

We now introduce the main object of our study.

\begin{definition}[Integrality gap preserving reduction (IGPR)]\label{def:gap_preserving_reduction}
Let $W \subseteq V \subseteq S_0$. A gap-preserving reduction is a function $f : V \to W$ such that for every instance $I \in V$,
\[
\ig(I) \leq \ig(f(I)).
\]
We say that $V$ gap-reduces to $W$ via $f$. In notation, $V \to_{f} W$.
\end{definition}

In informal terms, a gap-preserving reduction maps instances to a more restricted class without decreasing the integrality gap. This can be achieved in a few different ways:
\begin{lemma}\label{lem:gap_change}
    Let $I, I' \in S_0$. If one of the following conditions holds, then $\ig(I) \le \ig(I')$.
    \begin{enumerate}
        \item \(\ip(I)\le \ip(I')\) and \(\lp(I)\ge \lp(I')\),
        \item \(\ip(I') = \ip(I) - \delta\) and \(\lp(I') = \lp(I) - \varepsilon\), with $0 \le \delta \le \varepsilon$.
    \end{enumerate}
\end{lemma}

\begin{proof}
    The first case is trivial. For the second case, observe that
    \[
    \frac{\ip(I)}{\lp(I)} \le \frac{\ip(I)-\delta}{\lp(I)-\varepsilon} \Longleftrightarrow \frac{\delta}{\varepsilon} \le \frac{\ip(I)}{\lp(I)},
    \]
    and since $\ip(I)/\lp(I)\ge 1$ by definition, the statement follows.
\end{proof}
Lemma \ref{lem:gap_change} essentially allows us to prune or ``clean'' an instance, provided that the decrease in the fractional cost is at least as large as the decrease in the integer cost. The high-level idea is to construct a sequence of nested sets $S_0 \supseteq S_1 \supseteq \cdots \supseteq S_k$ along with reductions $f_{q+1} : S_q \to S_{q+1}, \; q \in \{0\} \cup [k - 1]$, where $[x] := \{1,2, \ldots, x\}$ for some $x \in \enne$, so that the final set $S_k$ consists of instances structured sufficiently to allow a direct analysis of the gap.

\begin{definition}[IGPR-chain]\label{def:igpr_chain}
    A chain of integrality gap preserving reductions (IGPR-chain) is a tuple made up of a sequence of subsets $S_0 \supseteq S_1 \supseteq \ldots \supseteq S_k$, and a sequence of functions $f_1, \ldots, f_k$ such that $S_q \to_{f_{q+1}} S_{q+1}$ for each $q=0, \ldots, k-1$. The length of the chain is $k$.
\end{definition}

In practice, the reductions will heavily exploit polyhedral characterizations of vertices. Starting from a vertex of $\lp(I)$, we iteratively modify $I$ such that the corresponding new vertex has even stronger properties as the one before, and the gap does not decrease.

\section{Weighted Vertex Cover}\label{sec:vertex_cover}
The \emph{Weighted Vertex Cover} problem (WVC) is defined as follows: given a graph $G=(V, E)$ and a non-negative node-weight function $w : V \to \erre_{>0}$, find a set of vertices $U\subseteq V$ such that $\{u,v\} \cap U \ne \emptyset$ for every edge $uv \in E$, and $w(U)=\sum_{v \in U} w_v$ is overall minimized. We require $w_v > 0$ for each node $v \in V$, or else $v$ could always be selected. The set of all instances is $S_0 := \{((A_G, I_n),(1_m, 0_n), w)\,|\,G$ is a node-weighted graph on $n$ nodes, $w\in \erre_{>0} ^n, n \in \enne \}$. For simplicity, an instance will be denoted by $(G, w)$.

The IP formulation \eqref{ip} can be rewritten in the following simpler way: 
\begin{subnumcases}{\ip(G,w):= \min \sum\limits_{v\in V} w_v x_v \,\,\,\text{ s.t.} }
    x_u + x_v \ge 1, & $uv \in E$, \label{edge_const} \\ 
    x_v \in \{0,1\}, & $v \in V$, \label{pos}
\end{subnumcases}
and LP$(G, w)$ is derived by replacing \eqref{pos} with  $1 \ge x_v \geq 0$ for all $v \in V$. The integrality gap is defined as $\ig(G,w)=\frac{\ip(G,w)}{\lp(G,w)}$, the global gap $\ig(WVC)$ is defined by \eqref{Gap} and is known to be $2$ \cite{nemhauser1974properties}; we provide an alternative, unified proof via an IGPR-chain.

As we hinted before, gap reductions exploit strong polyhedral properties of vertices of the LP-relaxation. For vertex cover, this property is the half-integrality of the corresponding polyhedron \eqref{edge_const}--\eqref{pos}; let us denote it by $\lp(G)$. Nemhauser and Trotter \cite{nemhauser1974properties} state that for any vertex $x^*$ of $\lp(G)$, $x^*_v \in \{0, \frac{1}{2}, 1\}$ holds for each node $v \in V$. A well-known result relates this observation to the concept of \emph{crown decompositions} in graphs \cite{Chor05}, which is a tri-partition $V=(C,H,B)$ such that $1)$ $C$ is an independent set, $2)$ there is no edge between $C$ and $B$, and $3)$ there exists a matching between $C$ and $H$ that covers every node in $H$. In particular, if we let $C(x^*)=\{v \in V: x^*_v = 0\},\,\, H(x^*)=\{v \in V: x^*_v = 1\},\,\, B(x^*)=\{v \in V: x^*_v = \frac{1}{2}\}$ for an optimal vertex $x^*$ of $\lp(G,1_n)$, the partition $(C(x^*), H(x^*), B(x^*))$ is a crown decomposition of $G$. 
A straightforward consequence of this connection is that a minimal vertex cover of $G$ can be obtained by taking the union of $H(x^*)$ and a minimal vertex cover of $B(x^*)$; a property which is heavily exploited in many practical applications and fixed parameter tractable algorithms.

If we replace $1_n$ with an arbitrary $w \in \erre_{>0}^n$, properties $1)$ and $2)$ hold automatically. Property $3)$ is replaced by a weaker one: each node $v\in H(x^*)$ has a neighbour in $C(x^*)$, or else we could decrease $x^*_v$ from $1$ to $\frac{1}{2}$ while strictly decreasing the total weight (recall $w_v > 0$ holds for all nodes $v \in V$). We build out first reduction based on this weaker crown-decomposition. In particular, we show in Lemma \ref{lem:vc_1} that restricting ourselves to the induced subgraph $G[B(x^*)]$ preserves the integrality gap. We can iteratively apply the same reasoning until only the trivial crown decomposition exists, corresponding to the single optimal vertex $x^* \equiv \frac{1}{2}$ of $\lp(G)$. From here, it follows that adding edges to the current graph maintains the gap, allowing us to focus solely on complete graphs (Lemma \ref{lem:vc_2}). We finally assert the gap is at most $2$ for such instances, proving the following theorem:

\begin{restatable}{theorem}{thmvc}\label{thm:ig_vertex_cover_final}
    $\ig(WVC)=2$.
\end{restatable}

Precise definitions and full details are provided in Appendix \ref{app:vertex_cover}. A summary of key steps can be found in Table \ref{tab:vc_summary}.

\begin{table}[ht]
    \centering
    \begin{tabular}{|c|c|c|}
         \hline
         Subset & Defining property & Reduction from $S_{i-1}$ \\
         \hline\hline
         $S_0$ & All instances $(G,w)$ & -- \\
         \hline
         $S_1$ & $(G,w)$ s.t. $x\equiv \frac{1}{2}$ is the only opt. sol. to $\lp(G,w)$ & Crown reduction \\
         \hline
         $S_2$ & $(K_n,w)$ in $S_1$ & Adding all missing edges \\
         \hline
    \end{tabular}
    \caption{Summary of the reductions for WVC. $\ig(S_2)=2$ is asserted manually.}
    \label{tab:vc_summary}
\end{table}

\section{Multi-Knapsack}\label{sec:multi_knap}

In the one-dimensional $0-1$ multiple knapsack problem (\emph{multiple knapsack} or \emph{multi-knapsack problem}), an input instance is encoded by a triple $(C,w,p)$, where $C=(C_1, \ldots, C_m) \in \enne_{>0} ^m$ is the capacity vector of $m$ knapsacks and $p=(p_1, \ldots, p_n) \in \enne_{>0}^n$, $w=(w_1, \ldots, w_n) \in \enne_{>0}^n$ are the profit and weight vectors of $n$ items. The problem consists of selecting a subset of items to be packed in each knapsack (an item can be chosen at most once) such that each knapsack's capacity constraint is respected and the combined profit of selected items is maximized. When $m$ is fixed, which is the case we are considering in this paper, the problem (denoted as $MK_m$) is weakly $\textsf{NP}$-hard and admits an FPTAS \cite{ibarra_kim,lawler_fptas}.

The classical integer linear formulation has a variable $x_{j,i}$ for each item $j \in [n]=\{1, \ldots, n\}$ and each knapsack $i\in [m]=\{1, \dots, m\}$.
\begin{subnumcases}{\ip(C,w,p):= \max \sum\limits_{j \in [n]}\sum\limits_{i \in [m]} x_{j,i} \cdot p_j \,\,\,\text{ s.t.} }
    \sum_{j \in [n]} x_{j,i} \cdot w_j \leq C_i, & $\forall i\in [m]$, \label{knap_const} \\ 
    \sum_{i \in [m]} x_{j,i} \le  1, & $\forall j\in [n]$, \label{item_const} \\
    x_{j,i} \in \{0,1\}, & $\forall j \in [n], \forall i \in [m]$. \label{pos_const_ks}
\end{subnumcases}
Setting $x_{j,i}=1$ corresponds to placing item $j$ in knapsack $i$. The constraint \eqref{knap_const} ensures compliance with knapsack capacities, whereas \eqref{item_const} reflects that each item can be selected for at most one knapsack. In the linear relaxation $\lp(C,w,p)$, constraint \eqref{pos_const_ks} is replaced by $x_{j,i}\ge 0$ for each $j\in[n]$ and $i\in[m]$. For convenience, both the linear relaxation and its optimum are denoted by $\lp(C,w,p)$. Since the problem is a maximization one, the definition in \eqref{gap} is inverted to consistently have gap values at least one: $\ig(C,w,p):= \lp(C,w,p)/\ip(C,w,p)$. Even though $MK_m$ is an ``easy'' optimization problem admitting an FPTAS, the above natural formulation has a gap of $m+1$ \cite{martello_toth}. A corresponding $(m+1)$-approximation algorithm exploits the fact that a special optimal solution of $\lp(C,w,p)$ (denoted as $x^*_s$) can be found by sorting items as $p_1 / w_1 \ge \ldots \ge p_n/w_n$, and greedily packing them into the current knapsack while moving excessive parts to the next one.

We build our reductions on this observation as well. A first preprocessing phase (Lemmas \ref{lem:mk_1}, \ref{lem:mk_2} and \ref{lem:mk_3}) clears the input from unnecessary knapsacks and items, guaranteeing that each knapsack is filled completely in $x^*_s$, and no items are ``hanging out'' from the last knapsack. The key reduction step comes next (Lemma \ref{lem:mk_4}), in which (almost) all pairs of items are merged iteratively that would otherwise fit inside the largest knapsack. In the resulting instance, any feasible integer solution contains at most one item in each knapsack, potentially with the exception of one single knapsack containing the pair $(n,j)$ for some $j \in [n-1]$. A follow-up step (Lemma \ref{lem:mk_5}) strengthens the previous one so that at most two knapsacks contain items in the integer optimal assignment. The last reduction phase (Lemma \ref{lem:mk_6}) ensures exactly one item is used in any optimal assignment, eliminating the remaining edge cases from the previous step caused by item $n$. The final instance space allows a manual calculation of the respective gap value, yielding the following theorem:

\begin{restatable}{theorem}{thmmk}\label{thm:mulit_knap}
    $\ig(MK_m)=m+1$.
\end{restatable}

Precise definitions and full details are provided in Appendix \ref{app:multi_knap}. A summary of key steps is presented in Table \ref{tab:mk_summary}.

\begin{table}[ht]
    \centering
    \begin{tabular}{|c|c|c|}
         \hline
         Subset & Defining property & Reduction from $S_{i-1}$ \\
         \hline\hline
         $S_0$ & All instances $(C,w,p)$ & -- \\
         \hline
         $S_1$ & No empty knapsack in $x^*_s$ & Deleting empty knapsacks \\
         \hline
         $S_2$ & No unassigned item in $x^*_s$ & Deleting unassigned items \\
         \hline
         $S_3$ & All knapsacks full in $x^*_s$ & $w':= w \cdot \nicefrac{C_{\text{sum}}}{w_{\text{sum}}}$ \\
         \hline
         $S_4$ & $\forall j_1, j_2\ne n: w_{j_1}+w_{j_2} > C_1$ & Merging items \\
         \hline
         $S_5$ & $\forall \sigma_{opt},\forall i>1: \sigma^{-1}_{opt}(i)\subseteq \{n\}$ & $C'_i:=C_i- w_j, w'_j := 0$ \\
         \hline
         $S_6$ & $\sigma_{opt}^{-1}(1)=\{j\}$, $\sigma_{opt}^{-1}(i)=\emptyset$ else & Case-dependent \\
         \hline
    \end{tabular}
    \caption{Summary of the reductions for $MK_m$. Properties are maintained when moving downward. $\ig(S_6)=m+1$ is asserted manually.}
    \label{tab:mk_summary}
\end{table}

\section{Unrelated machine scheduling}\label{sec:machine_sched}

In the \emph{unrelated parallel machine scheduling} problem (\emph{machine scheduling} for short), an input instance is given by a matrix $P=(p_{j,i})\in \enne_+^{n \times m}$, and represents a collection of $n$ jobs labelled by $[n]=\{1,\ldots, n\}$ that need to be assigned on $m$ machines $[m]=\{1,\ldots, m\}$, where job $j$ has a processing time $p_{j,i}\in \enne_+$ on machine $i$. The goal is to minimize the \emph{makespan} of a schedule $\sigma: [n]\to [m]$, defined as $C_{max}(\sigma,P)=\max\{C_{\sigma,P}(i): i \in [m]\}$, where $C_{\sigma,P}(i)=\sum\limits_{j \in \sigma^{-1}(i)} p_{j,i}$ is the \emph{completion time} of machine $i$ according to the assignment $\sigma$. It is well-known that $R||C_{max}$ (notation of \cite{machine_scheduling_review}) cannot be approximated better than $\frac{3}{2}$ \cite{lenstra_shmoys_tardos}, implying an integrality gap of at least $\frac{3}{2}$ for any polynomially solvable relaxation.

For the sake of our reductions, we consider a slightly more general problem in which some jobs $j$ cannot be assigned on some machines $i$. It is equivalent to setting $p_{j,i}=\infty$, which renders $C_{\sigma,P}$ to be $\infty$ if the sum contains an infinite element. For a job $j$, the set $\cM_j = \{i \in [m]: p_{j,i}\ne \infty\}$ contains machines on which $j$ can be placed. We further require each row and column of $P$ to have at least one non-infinite entry, eliminating degenerate cases where a machine cannot accommodate any jobs or a job cannot be placed on any of the machines. The problem, known as \emph{unrelated restricted assignment problem} and denoted as $R|\cM_j|C_{max}$, has the same approximation guarantees as its special case $R||C_{max}$.

The classic IP-formulation of $R|\cM_j|C_{max}$ has a variable $x_{j,i}$ for each pair $(j,i)$ such that $p_{j,i}\ne \infty$.
\begin{subnumcases}{\ip(P):= \min T \,\,\,\text{ s.t.} }
    \sum_{j: i \in \cM_j} x_{j,i}\cdot p_{j,i} \leq T, & $\forall i\in [m]$, \label{makespan_const} \\ 
    \sum_{i \in \cM_j} x_{j,i} = 1, & $\forall j\in [n]$, \label{job_const} \\
    x_{j,i} \in \{0,1\}, & $\forall j \in [n], \forall i \in \cM_j$. \label{pos_const}
\end{subnumcases}
Constraint \eqref{makespan_const} and the objective function ensure the maximal completion time of machines is minimal, while \eqref{job_const} guarantees each job gets assigned to exactly one machine. The linear relaxation $\lp(P)$ is obtained by replacing \eqref{pos_const} with $x_{j,i}\ge 0$ for all $j \in [n], i \in \cM_j$. For convenience, we denote by $\lp(P)$ both the linear program and its optimal value. The formulation has an unbounded integrality gap (see Introduction); in particular the gap restricted to $m$ machines is exactly $m$.

Vertices of a parametrized version of $\lp(P)$ have been studied in \cite{lenstra_shmoys_tardos} and \cite{Vazirani03} (chapter $17.3$), with the help of an auxiliary bipartite graph $G(x^*)$ that encodes fractional variables of an $\lp$-optimal solution $x^*$. We repeat their arguments and prove in Lemma \ref{lem:vert_char} that there exists an optimal vertex of $\lp(P)$ such that $G(x^*)$ is a forest. Building on this, we construct an IGPR-chain that can be divided into three main conceptual phases. In the first three steps, (Lemmas \ref{lem:ms_1}, \ref{lem:merge_int} and \ref{lem:graph_conn}), we perform a preliminary cleaning of the instance—eliminating unused assignments—to ensure connectivity of the graph, and we consolidate integer jobs. Second (Lemmas \ref{lem:int_job} and \ref{lem:ms_5}), we structurally simplify the leaves of $G(x^*)$, which by now is known to be a tree. Finally (Lemmas \ref{lem:ms_6} and \ref{lem:ms_7}), we iteratively ``prune'' the tree around its leaves—rebalancing jobs and deleting machines—until only a single fractional job remains. Throughout this process, the recursive application of the chain ensures that local simplifications do not violate the global structural properties obtained prior to that point. Eventually, we prove

\begin{restatable}{theorem}{thmms}\label{thm:machine_sched}
    $\ig(R(\le M)|\cM_j|C_{max}) = M$, $\ig(Rm||C_{max})=m$.
\end{restatable}

Precise definitions and full details are provided in Appendix \ref{app:machine_sched}. A summary of key steps is presented in Table \ref{tab:ms_summary}.

\begin{table}[ht]
    \centering
    \begin{tabular}{|c|c|c|}
         \hline
         Subset & Defining property & Reduction from $S_{i-1}$ \\
         \hline\hline
         $S_0$ & All instances $P$ & -- \\
         \hline
         $S_1$ & $x^*_{j,i}=0 \Rightarrow p_{j,i}=\infty$ & Increasing $p_{j,i}$ to $\infty$ \\
         \hline
         $S_2$ & $\le 1$ integer job on each $i$ & Merging jobs \\
         \hline
         $S_3$ & $G(P)$ is a tree & Restricting to a component \\
         \hline
         $S_4$ & Leaf-machines have no int. $j$ & $p'_{j_{int},i}=0$, $p'_{j_{frac},i}=p_{j_{frac},i} + p_{j_{int},i}$ \\
         \hline
         $S_5$ & $p_{j,i} = \ip(P)$ $\forall$ frac. $j$ $\forall$ leaf $i$ & $p_{j,i}:= \ip(P)$ \\
         \hline
         $S_6$ & One frac. $j$ or $u$ has no int. $j$ & $p'_{j,u}=0$, $p'_{v_d,u}=p_{v_d,u} + p_{j,u}$ \\
         \hline
         $S_7$ & $\exists!$ frac. $j$ & $p'_{v_d, u}=0, p'_{w,u}=p_{v_d,u} + p_{w,u}$ \\
         \hline
    \end{tabular}
    \caption{Summary of the reductions for $R(\le M)|\cM_j|C_{max}$. Properties are maintained when moving downward. $\ig(S_7)=M$ is asserted manually.}
    \label{tab:ms_summary}
\end{table}

\section{Restricted assignment problem via the configuration LP}\label{sec:conifg_lp}

A fundamental goal of the scheduling community is to find an IP-formulation with the best theoretically possible gap of $\frac{3}{2}$. The \emph{configuration} LP (CLP) was proposed as an alternative to the natural formulations by Bansal and Sviridenko \cite{Bansal06}; and while some special cases (including $R||C_{max}$) already admit a gap of $2$ \cite{Ebenlendr14,Verschae14}, improvements have been made for the (identical) restricted assignment framework \cite{Jansen17-2, Jansen17, Svensson11} (in which a job $j$ has the same processing time $p_j\in \enne_+$ on all machines $i\in \cM_j$), placing $\ig(P|\cM_j|C_{max})$ between $\frac{3}{2}$ and $\frac{11}{6}$. In the CLP formulation, a parameter $T$ is introduced as a proxy for minimizing the makespan, and instead of assigning individual jobs, it considers sets of jobs called \emph{configurations}. Given an instance $P$ with $n$ jobs and $m$ machines, let $\cC_i(P,T)$ denote the family of configurations whose total processing time on machine $i$ is at most $T$:
\begin{equation}
    \label{config} \cC_i(P,T) := \left\{C \subseteq [n]: \sum\limits_{j \in C} p_{j,i} \le T\right\}. 
\end{equation}
The linear program uses a variable $x_{i,C}$ for each configuration $C$ and machine $i$, and is defined as follows:
\begin{subnumcases}{\clp(P,T):} 
    \sum\limits_{C \in \cC_{i}(P,T)} x_{i, C} = 1, & $i \in [m]$, \label{machine_const}\\ 
    \sum\limits_{i \in [m]} \sum\limits_{C \in \cC_i(P,T): j \in C} x_{i,C} =1, & $j \in [n]$, \label{ass_const}\\ 
    x_{i,C} \ge 0, & $i \in [m], \, C \in \cC_i(P,T)$. \label{nonneg_const} 
\end{subnumcases}

Let $\lp(P)$ denote the smallest integer value of $T$ for which $\clp(P,T)$ is feasible, and let $\ip(P)$ denote the optimal makespan of the integer problem. It is immediate that $\lp(P) \le \ip(P)$: constraint~\eqref{machine_const} ensures that each machine selects exactly one configuration, while~\eqref{ass_const} guarantees that each job is assigned exactly once. The integrality gap is given as $\ig(P) = \ip(P) / \lp(P)$.

We present further evidence for the strength of the CLP relaxation by proving that it admits a gap of $1$ for some ``easy'' subcases of $P|\cM_j|C_{max}$. Our first result concerns the two-machine case, denoted as $P2|\mathcal{M}_j|C_{\max}$.

\begin{restatable}{theorem}{thmmequaltwo}\label{thm:p2}
    $\ig(P2|\mathcal{M}_j|C_{\max})=1$.
\end{restatable}

As preprocessing, we gap-reduce the problem to the restriction-free case $P2||C_{max}$ (Lemmas \ref{lem:exclusive}, \ref{lem:clp2_2} and \ref{lem:clp2_3}). For this simpler version, we show the optimal integer makespan coincides with a known lower bound (Lemma \ref{lem:iden}), asserting the gap is $1$. Precise definitions and full details are provided in Appendix \ref{app:clp2}. A summary of key steps is presented in Table \ref{tab:config2_summary}. 

\begin{table}[ht]
    \centering
    \begin{tabular}{|c|c|c|}
         \hline
         Subset & Defining property & Reduction from $S_{i-1}$ \\
         \hline\hline
         $S_0$ & All instances $P$ & -- \\
         \hline
         $S_1$ & $\le 1$ exclusive job on both machines in $P$ & Merging exclusive jobs \\
         \hline
         $S_2$ & $\le 1$ exclusive job in $P$ & $p_{j_1},p_{j_2} -= p_{j_1}$ \\
         \hline
         $S_3$ & $P2||C_{max}$ & $p_{j_1,1} = \infty \to p_{j_1,2}$. \\
         \hline
    \end{tabular}
    \caption{Summary of the reductions for $P2|\cM_j|C_{max}$. $\ig(S_3)=1$ is asserted manually.}
    \label{tab:config2_summary}
\end{table}

Note that Theorem \ref{thm:p2} is not replicable for $m>2$ machines, as shown by the instances in \cite{Kurpisz18} (pp $235-236$) indicating $\ig(P3||C_{max})\ge \frac{1024}{1023}$. This warns us that having a PTAS might not be sufficient enough to guarantee gap values of $1$ ($P||C_{max}$ admits a PTAS, see \cite{Encz25}). However, solvability in polynomial time is an even stronger indicator that the gap should equal $1$; we assert this for known such cases of $P|\mathcal{M}_j|C_{max}$.

Our second result considers the case where all processing times are equal. This setting is commonly denoted as $P|\mathcal{M}j, p_j = 1|C_{\max}$, since one can assume without loss of generality that $p_j = 1$ for all $j \in [n]$. Lin and Li \cite{Lin04} constructed a polynomial-time algorithm by modelling it as a network flow problem. We complement their result by proving CLP has a gap of $1$ relying on an IGPR-chain.

\begin{restatable}{theorem}{thmunit}\label{thm:all_one}
$\ig(P|\mathcal{M}_j, p_j = 1|C_{\max}) = 1$.
\end{restatable}

The key underlying idea is to study the ``local modification digraph'' corresponding to an optimal integer assignment $\sigma$. Nodes represent machines, and an arc $\overrightarrow{\sigma(j),i}$ reflects that job $j$ is allowed on machine $i$ as well as on $\sigma(j)$. In Lemma \ref{lem:graph_connect}, we gap-reduce to instances where the digraph has certain connectivity properties. This connectivity guarantees that jobs are placed evenly in $\sigma$ (Lemma \ref{lem:graph_moving}), entailing the gap is $1$. Key steps are described in detail in Appendix \ref{app:unit}.

The last case we consider is the \emph{graph balancing} problem $P|\mathcal{M}_j(\le 2)|C_{\max}$, in which each job can be placed on at most $2$ machines. The name arises from the following interpretation. Given an instance $P$, construct an edge-weighted graph $G_P=(V,E)$ where vertices correspond to machines. Each job $j$ that can be assigned to two machines $m_{j,1}$ and $m_{j,2}$ is represented by an edge $e_j=(m_{j,1}, m_{j,2})$ with weight $w(e_j)=p_j$. If a job can only be assigned to a single machine $m_j$, it is represented by a loop $(m_j, m_j)$. Assigning jobs corresponds to orienting edges, and the completion time of a machine corresponds to the weighted in-degree of its node. The problem can thus be reformulated as finding an orientation minimizing the maximum weighted in-degree. 

The integrality gap of this problem is know to lie between $\nicefrac{3}{2}$ \cite{Asahiro11} and $1.749$ \cite{jansen2018local}. In the slightly more general \emph{uniform graph balancing} ($Q|\mathcal{M}_j (\le 2)|C_{max}$; jobs are identical but each machine $i$ has a designated speed $s_i$), the gap is already $2$ \cite{Ebenlendr14}. 
When $G_P$ is a tree, a simple observation guarantees an integer assignment matching the trivial lower bound $\max_j p_j$: if we orient $G_P$ to be a rooted tree, each node will be the head of at most one arc, hence the makespan will be exactly $p_{j_{max}}$. This observation translates to $\clp$ as well, since there has to be at least one configuration containing $j_{max}$ on some machine, thus $\lp(P)\ge \max_j p_j$ holds.

\begin{restatable}{proposition}{thmgraphbalancing}
For an instance $P$ of $P|\mathcal{M}_j (\le 2)|C_{max}$, $\ig(P)=1$ if $G(P)$ is a tree.
\end{restatable}

\section{Conclusion and open problems}\label{sec:conclusion}

In this paper, we presented a novel technique called integrality gap-preserving reduction. This theoretical framework is built on the empirical observation that instances maximizing the integrality gap are highly structured. By demonstrating that one only needs to analyse a well-defined subset of instances to prove the global integrality gap, we validated this methodology across three well-known optimization problems: the weighted vertex cover problem, the unrelated machine scheduling problem, and the multiple knapsack problem. Furthermore, we showed the method is capable of advancing the current knowledge regarding the integrality gap of the restricted assignment problem with respect to its configuration LP relaxation. In most of these cases, we relied on a polyhedral study of the corresponding linear relaxations, starting from prior knowledge on vertices of the polytopes. It is exactly the lack of this understanding of vertices that prohibited us from carrying out a deeper analysis of the configuration LP, in contrast to the strong properties we heavily exploited in the other problems. This gave motivation to further study the properties of the CLP polytope, leading us to the following conjecture, which we experimentally asserted for several small-scale instances.

\begin{conjecture}\label{conj:gb}
    For every instance $P$ of the restricted assignment problem, the polyhedron $\clp(P,\lp(P))$ contains a half-integral feasible solution; that is, a vector $x$ such that $x_{i,C}\in \{0, \frac{1}{2}, 1\}$ for all $i \in [m]$ and $C \in \cC_{i}(P,\lp(P))$.
\end{conjecture}

Note that an equivalent form requires a half-integral solution of $\clp(P, T)$ for every $T \ge \lp(P)$, where equivalence follows from the fact that $\clp(P,\lp(P)) \subseteq \clp(P,T)$ for such $T$. Further note that we do not require $x$ to be a vertex of $\clp(P,\lp(P))$, as optimality is sufficient for constructing gap-reductions. Initially, we formulated a stronger version of the conjecture, which stated the half-integrality for each vertex of the polyhedra, similar to vertex cover. However, this stronger version was disproved experimentally, while the current version has resisted such refutations so far.

In a half-integral solution of $\clp(P,\lp(P))$, each job appears in at most two configurations. These correspond to at most two distinct machines, and all machines apart from this pair can be safely excluded, by setting the corresponding entries of $P$ to $\infty$. It follows that every general instance can be transformed into another one in which $|\mathcal{M}_j| \le 2$ while preserving the gap. Consequently, if Conjecture \ref{conj:gb} holds, the integrality gap of the configuration LP reduces to that of the graph balancing problem, and it would follow immediately that $\ig(P|\cM_j|C_{max}) \le 1.749$ by \cite{jansen2018local}.

\bibliography{biblio}

\appendix

\section{Weighted Vertex Cover}\label{app:vertex_cover}

Let $S_1 \subseteq S_0$ denote the instances $(G,w)$ where $x\equiv \frac{1}{2}$ is the only optimal solution of $\lp(G,w)$. For $(G,w)\in S_1$, define $f_1(G,w)=(G,w)$. Else, the face $\lp(G)\cap \{x: w(x)=\lp(G,w)\}$ contains a vertex $x^*$ of $\lp(G)$ different from $\frac{1}{2}$ that is optimal for $w$. Let $G':= G[B(x^*)]$ and $w':= w|_{B(x^*)}$ where $(C(x^*),H(x^*),B(x^*))$ is the crown-decomposition corresponding to $x^*$. Since $B(x^*)\ne V$, the decomposition is non-trivial. If $(G',w') \not\in S_1$, repeat the process until the resulting instance only has $\frac{1}{2}$ as an optimal solution, or $V'=\emptyset$. Let $f_1 (G,w)$ be this final instance, which must exist as $|V|$ strictly decreases in each iteration.

\begin{lemma}\label{lem:vc_1}
    $S_0 \to_{f_1} S_1$.
\end{lemma}

\begin{proof}
It suffices to prove that one crown reduction preserves the gap. Let $(G,w)\in S_0 \setminus S_1$, and let $x^*$ be an optimal vertex of $\lp(G,w)$ different from $\frac{1}{2}$. Consider the following auxiliary modification of $w$:
\[
w^*_v := \begin{cases}
    \sum_{u \not\in C(x^*)} {w_u}+1 & v \in C(x^*) \\
    w_v                         & \text{otherwise.}
\end{cases}
\]
The inequalities $\ip(G,w)\le \ip(G,w^*)$ and $\lp(G,w)\le \lp(G,w^*)$ hold, since $w\le w^*$ coordinate-wise. On the other hand, $w(x^*)=w^*(x^*)$ imply $\lp(G,w)=\lp(G,w^*)$. An optimal solution of $\ip(G,w^*)$ cannot contain $v \in C(x^*)$, as $H(x^*)\cup B(x^*)$ is a vertex cover with strictly smaller weight. Thus, as in the $w=1_n$ case, a minimum-weight vertex cover $U$ of $(G,w^*)$ must contain $H(x^*)$ to cover edges between $C(x^*)$ and $H(x^*)$ (recall that each $v \in H(x^*)$ is needed, as each of them has a neighbour in $C(x^*)$). These already cover every edge inside $H(x^*)$ and between $H(x^*)$ and $B(x^*)$, so the rest of $U$ is a minimum-weight vertex cover in the induced subgraph $G[B(x^*)]$ with respect to $w^*|_{B(x^*)}=w|_{B(x^*)}$. With $G' = G[B(x^*)]$ and $w'=w|_{B(x^*)}$, it holds that $\ip(G,w^*) = w(H(x^*))+ \ip(G',w')$ and $\lp(G,w^*)= 0 \cdot w^*(C(x^*))+1 \cdot w(H(x^*))+ \frac{1}{2}\cdot w(B(x^*))$. Note that $G'$ does not have isolated vertices: such a vertex $v$ must be the $G$-neighbour of some node in $H(x^*)$, but then $x^*_v:= 0$ would strictly decrease the weight of $x^*$. Hence, $x\equiv \frac{1}{2}$ is feasible for $\lp(G')$, implying $\lp(G',w')\le \frac{1}{2}\cdot w(B(x^*))$ and $\lp(G,w^*)\ge w(H(x^*)) + \lp(G',w')$. Thus,

\begin{equation}\label{vc}
    \frac{\ip(G',w')}{\lp(G',w')} \ge \frac{\ip(G,w^*)-w(H(x^*))}{\lp(G,w^*)-w(H(x^*))} \ge \frac{\ip(G,w^*)}{\lp(G,w^*)} \ge \frac{\ip(G,w)}{\lp(G,w)},
\end{equation}
where the last two inequalities follow from Lemma \ref{lem:gap_change} Case 2 and Case 1, respectively. Observe that \eqref{vc} is true even if $V(G') = \emptyset$ (with the convention that $\ig(\emptyset,w')=1$): $B(x^*) = \emptyset$ implies $\ip(G,w^*)=\lp(G,w^*)=w(H(x^*))$.
\end{proof}

The next reduction builds on the fact that if $x\equiv \frac{1}{2}$ is the only optimal solution of $\lp(G,w)$, then it is the only optimal solution to $\lp(G\cup\{e\}, w)$ (where $e \not \in E(G)$) as well. Hence, adding an edge to $G$ does not change $\lp(G,w)$, but it may increase $\ip(G,w)$, so $\ig(G,w)\le \ig(G\cup\{e\}, w)$ is true. Let us define $S_2 =\{(K_n, w) \in S_1: w \in \erre_+ ^n, n \in \enne\}$, where $K_n$ is the complete graph on $n$ vertices. Note that not every $(K_n, w)$ belongs to $S_2$, e.g. $n=3, w_u=2, w_v =1, w_s=1$ has the optimal solution $x_u = 0, x_v=x_s=1$. Define $f_2: S_1 \to S_2$ as $f(G,w)=(K_n,w)$ where $G$ has $n$ vertices. The following lemma is immediate.

\begin{lemma}\label{lem:vc_2}
    $S_1 \to_{f_2} S_2$.
\end{lemma}

As anticipated, a strong enough structure in the reduced instance space enables a manual calculation of the gap. Assume $(K_n, w)\in S_2$ and $w_1 \le \ldots \le w_n$. Then $\ip(K_n, w)=w_1 + \ldots + w_{n-1}$, as every vertex cover consists of at least $n-1$ nodes. On the other hand, $\lp(K_n, w)=\frac{w_1 + \ldots + w_n}{2}$, as $x \equiv \frac{1}{2}$ is the only optimal fractional solution.


\thmvc*

\begin{proof}
    $((S_0, S_1, S_2), (f_1, f_2))$ is an IGPR-chain, and
    \[
    \ig(S_2) = \sup\left\{2\cdot \frac{w_1 + \ldots + w_{n-1}}{w_1 + \ldots + w_n}: w_1 \le \ldots \le w_n, w \in \erre_+^n, n \in \enne\right\}=2,
    \]
    where the tightness of the supremum is shown by the $w \equiv 1$ instances.
\end{proof}

\section{Multi-Knapsack}\label{app:multi_knap}

Since $MK_m$ is a maximization problem and the definition of the integrality gap is inverted, we also have to modify Lemma \ref{lem:gap_change}. Let $S_0$ be the instance set of a maximization problem $\cP$.

\begin{lemma}\label{lem:gap_change_max}
    Let $I, I' \in S_0$. If one of the following conditions holds, then $\ig(I)\le \ig(I')$.
    \begin{enumerate}
        \item $\lp(P)\le \lp(P')$ and $\ip(P) \ge \ip(P')$,
        \item $\lp(P')=\lp(P)-\delta$ and $\ip(P')=\ip(P)-\varepsilon$, with $0 \le \delta \le \varepsilon$.
    \end{enumerate}
\end{lemma}

\begin{proof}
    The proof is identical to that of Lemma \ref{lem:gap_change}.
\end{proof}

The classical $(m+1)$-approximation algorithm builds on two key observations. First, Martello and Toth (\cite{martello_toth}, Chapter 6.2.1) show that the linear relaxation $\lp(C,w,p)$ is in essence equivalent with merging the $m$ knapsacks into a larger one with capacity $\sum_{i \in [m]} C_i$, and then considering its natural relaxation $\lp(\sum_{i \in [m]} C_i, w,p)$. The second observation from Dantzig \cite{dantzig} shows that decreasingly sorting the items by their unit profit $p_i/w_i$ (with arbitrary tie-breaking) allows a greedy approach to find an optimal solution of the latter linear relaxation. The method iteratively processes the items in the sorted order, and assigns them integrally to the first knapsack. If a fraction of the current item does not fit entirely, then the excessive fractional part is moved to the next knapsack, and the process resumes from there. It follows from these two facts that this solution, denoted by $x^*_s$, is indeed optimal for $\lp(C,w,p)$. 

Let $f_1, \ldots, f_m \in [n]\cup\{\infty\}$ denote the (at most) $m$ items for which a new knapsack was opened in the process, with $f_k = \infty$ if the $k$-th item is not defined. Note that $f_{i}=\ldots = f_j$ is possible for any two indices $i$ and $j$. Item $f_k$ is referred to as the \emph{critical item in knapsack} $k$, and is obtained as $f_k:= \min\left\{j \in [n]: \sum\limits_{l=1}^j w_l > \sum\limits_{i=1}^k C_i\right\}$. Items not in $\{f_1, \ldots, f_m\}$ are assigned integrally in $x^*_s$, and therefore are called \emph{integer items}, whereas if $\sum\limits_{l=1}^{f_k-1} w_l < \sum\limits_{i=1}^k C_i$, then $f_k$ is split fractionally between knapsacks $k$ and $k+1$; else item $f_k$ is integrally assigned to knapsack $k+1$.

We focus on instances where the number of machines is at most a fixed constant $M>1$. This will prove advantageous later, as some reduction steps might decrease the number of knapsacks in the process. The set of instances therefore is $S_0:= \{(C,w,p)\in \enne_{>0}^{m+n+n}: m \le M, \max_j\{w_j\} \le \max_i \{C_i\}\}$, where the latter condition forbids degenerate inputs in which an item does not fit inside any of the knapsacks. Let us assume the items are sorted decreasingly by their unit profit: $\nicefrac{p_1}{w_1} \ge \ldots \ge \nicefrac{p_n}{w_n}$, with ambiguity resolved via an arbitrary tie-breaking. We may further assume without loss of generality that the first knapsack has the largest capacity: $C_1 \ge C_2, \ldots ,C_m$. Let us further denote $C_{\text{sum}}:= \sum_{i \in [m]} C_i$, $p_{\text{sum}}:= \sum_{j \in [n]} p_j$ and $w_{\text{sum}}:= \sum_{j \in [n]} w_j$. 

Our first reduction ensures that there are no empty knapsacks in $x^*_s$. Let $S_1 \subseteq S_0$ be the set of instances for which $f_{m-1}< \infty$. For an instance $(C,w,p)\in S_1$, let $f_1(C,w,p)=(C,w,p)$. Else, assume that the last critical item is $f_k$ with $k\le m-2$, and $f_{k+1}=\ldots = f_m = \infty$. Let $C'=(C_1, \ldots, C_{k+1})$, and let $f_1(C,w,p):= (C',w,p)\in S_1$.

\begin{lemma}\label{lem:mk_1}
    $S_0 \to_{f_1} S_1$.
\end{lemma}

\begin{proof}
    Given that $f_{k+1}=\ldots = f_m = \infty$, in the process of constructing $x^*_s$, all items after $f_k$ were assigned integrally to knapsack $k+1$, along with item $n$. Hence, knapsacks $k+2, \ldots, m$ are left empty, and we can get rid of them while maintaining the feasibility of $x^*$. At the same time, getting rid of knapsacks can only decrease $\ip(C,w,p)$, so the integrality gap can only increase.
\end{proof}

In the next preprocessing step, we ensure no item is hanging out completely of the last knapsack in $x^*_s$, by throwing away such unnecessary items. Let $S_2\subseteq S_1$ be the set of instances in $S_1$ where either $f_m = \infty$, or $f_m =n$. For an instance $(C,w,p)\in S_2$, let $f_2(C,w,p)=(C,w,p)$. Else, assume that $f_m < \infty$ but $f_m < n$. Let $w'=w|_{[f_m]}$ and $p'=p|_{[f_m]}$. Define $f_2(C,w,p)=(C,w',p')\in S_2$.

\begin{lemma}\label{lem:mk_2}
    $S_1 \to_{f_2} S_2$.
\end{lemma}

\begin{proof}
    The modification does not affect $x^*_s$, so $\lp(C,w,p)=\lp(C,w',p')$. On the other hand, throwing away items can only decrease $\ip(C,w,p)$, so $\ip(C,w,p) \ge \ip(C,w',p')$ holds. The statement follows from Lemma \ref{lem:gap_change_max}.
\end{proof}

As a last preprocessing, we ensure that even knapsack $m$ is filled entirely. Let $S_3 \subseteq S_2$ be the set of instances in $S_2$ where $C_{\text{sum}} \le w_{\text{sum}}$. For $(C,w,p)\in S_3$, let $f_3(C,w,p)=(C,w,p)$. Else, if $w_{\text{sum}} < C_{\text{sum}}$, let $r:= \nicefrac{C_{\text{sum}}}{w_{\text{sum}}}>1$. Define $w'$ as $w'=r\cdot w=(r\cdot w_1, \ldots, r\cdot w_n)$. Let $f_3(C,w,p)= (C,w',p)$. Note that $w'_{\text{sum}} = C_{\text{sum}}$.

\begin{lemma}\label{lem:mk_3}
    $S_2 \to_{f_3} S_3$.
\end{lemma}

\begin{proof}
    Let $(C,w,p)\in S_2\setminus S_3$. On the one hand, $\lp(C,w,p)=\lp(C,w',p)=p_{\text{sum}}$, as all $n$ items can be packed fractionally in both cases. On the other hand, $\ip(C,w,p) \ge \ip(C,w',p)$, because we just increased item weights. The statement follows.
\end{proof}

Instances in $S_3$ are ``balanced'', in the sense that there are no unnecessary items or knapsacks that are not used in the LP-optimal solution. The profit of this fractional assignment can be calculated explicitly, since items $1,\ldots, n-1$ are included completely, and the outlying part of item $n$ accounts for the $(w_{\text{sum}}-C_{\text{sum}})$-fraction of its total profit.
\[
p(x^*_s)= p_1+\ldots + p_{n-1} + \frac{C_{\text{sum}}-(w_1 + \ldots w_{n-1})}{w_n}\cdot p_n,
\]
so
\begin{equation}\label{eq:mk_lpopt}
    \lp(C,w,p)= p_{\text{sum}} - (w_{\text{sum}}-C_{\text{sum}})\cdot \frac{p_n}{w_n}
\end{equation}
for each $(C,w,p)\in S_3$.

The key reduction in our chain exploits \eqref{eq:mk_lpopt} to merge items while maintaining the LP-optimum. Merging two items can only decrease the integer optimum, so the gap will be increased by such an operation. However, not all pairs of items can be merged unconditionally. First off, we need to maintain the property that each item can be packed into at least the largest knapsack (in formula, $\max_j\{w_j\}\le \max_i \{C_i\}$). Otherwise, merging items would lead to a territory of unbounded integrality gaps, possibly even reaching instances with infinite gap: if eventually we end up with one item that no longer fits inside any of the knapsacks, the integer optimum is $\infty$. Furthermore, in order for \eqref{eq:mk_lpopt} to remain unchanged by merging two items, they must be different from item $n$. These issues can be corrected easily by merging only specific pairs of items.

Let $S_4 \subseteq S_3$ be the set of instances in $S_3$ such that for each pair of items $j_1, j_2 \ne n$, it holds that $w_{j_1}+w_{j_2}> C_1$. For an instance $(C,w,p) \in S_4$, let $f_4(C,w,p)= (C,w,p)$. Else, assume $j_1, j_2 \ne n$ and $w_{j_1} + w_{j_2} < C_1$. First, let us multiply the profit vector $p$ by $\nicefrac{w_n}{p_n}$; let $p'$ be the resulting vector. It does not change the gap, since both the integer and fractional optima get multiplied by the same constant. Note that $\nicefrac{p'_1}{w_1}\ge \ldots \ge \nicefrac{p'_n}{w_n}=1$, and \eqref{eq:mk_lpopt} becomes
\begin{equation}\label{eq:mk_lpopt_one}
\lp(C,w,p') = p'_{\text{sum}} - w_{\text{sum}} + C_{\text{sum}}.
\end{equation}

Let us merge items $j_1$ and $j_2$ into an item $j_{new}$ with $p'_{j_{new}}:= p'_{j_1} + p'_{j_2}$ and $w_{j_{new}}:= w_{j_1} + w_{j_2}$. Let $(C,w'', p'')$ be the resulting instance. If there are still two items that can be merged (the sum of their weight is at most $C_1$, and they are different from $n$), repeat the process until the condition of $S_4$ holds. Let $f_4(C,w,p)$ be this final instance.

\begin{lemma}\label{lem:mk_4}
    $S_3 \to_{f_4} S_4$.
\end{lemma}

\begin{proof}
    It is enough to prove the gap-preserving property of one round in the reduction. Multiplying the vector $p$ modifies the integer and the linear optima by the same constant, so the gap does dot change. Clearly, $\ip(C,w,p')\ge \ip(C,w'',p'')$, as the integer feasible solutions for the latter can be perceived naturally as a subset of the integer solutions of the former. For the linear optimum, note that $p''_{\text{sum}}=p_{\text{sum}}$ and $w''_{\text{sum}}=w_{\text{sum}}$, and even though the sorted order of items might change from $(C,w,p')$ to $(C,w'',p'')$, item $n$ will remain the last one:
    \[
     \frac{p'_{j_{new}}}{w_{j_{new}}} = \frac{p'_{j_{1}} + p'_{j_{2}}}{w_{j_{1}} + w_{j_{2}}} \ge \frac{w_{j_{1}} + w_{j_{2}}}{w_{j_{1}} + w_{j_{2}}}=1= \frac{p'_n}{w_n},
    \]
    using the fact that $\nicefrac{p'_{j_1}}{w_{j_1}} \ge \nicefrac{p'_n}{w_n}=1$, so $p'_{j_1}\ge w_{j_1}$.

    Consequently, \eqref{eq:mk_lpopt_one} shows that 
    \[
    \lp(C,w'',p'') = \lp(C,w,p').
    \]
    The statement follows from Lemma \ref{lem:gap_change_max}, and from $\ig(C,w,p')=\ig(C,w,p)$.
\end{proof}

The property of $S_4$ gives a very strong structure to integer solutions. In particular, in any assignment, with the exception of one knapsack, all knapsacks contain either one or zero items, whereas the potential exception may contain a pair of items, among which $n$ must be present. We can exploit this property in the following way: let $i$ be an arbitrary knapsack with item $j\ne n$ assigned to it in an optimal integer assignment. If we decrease both $C_i$ and $w_j$ by $w_j$, then \eqref{eq:mk_lpopt} does not change, and by the property that no two items (except maybe item $n$ with another item) fit inside any knapsack together, the integer optimum will remain the same as well. The only caveat we have to bear in mind is that all items must fit inside the largest knapsack (or else the issue with a potentially infinite gap would re-emerge), so the modification cannot be done for the first knapsack. Moreover, in order for \eqref{eq:mk_lpopt} to be an invariant during the reduction, we cannot modify $w_n$, creating edge cases when item $n$ is involved in the optimal solution. We elaborate on the idea below.

Let $\sigma: [n]\to [m+1]$ denote an assignment of items to knapsacks, where $\sigma(j)=i$ means item $j$ is placed in knapsack $i$, and $\sigma(j)=m+1$ means item $j$ was not included in any of the knapsacks. The assignment $\sigma$ is called \emph{feasible} if $w(\sigma_i):= \sum_{j \in [n]: \sigma(j)=i} w_j \le C_i$ holds for each knapsack $i \in [m]$. Let $\sigma_{opt}$ be an optimal assignment for an instance $(C,w,p)\in S_4$, and let $i>1$ be a knapsack such that $\sigma_{opt}^{-1}(i) = \{j,n\}$ or $\sigma_{opt}^{-1}(i) = \{j\}$ for some $j \in [n-1]$. Let $C'_i:= C_i - w_j$, $w'_j := w_j - w_j =0$, and $C'_{i'}=C_{i'}$, $w_{j'}':= w_{j'}$ for all $i' \ne i, j' \ne j$. Let $p'\equiv p$, and dispose of item $j$ from the instance, giving $(C'',w'',p'')=(C'|_{[n]\setminus \{j\}}, w'|_{[n]\setminus \{j\}}, p'|_{[n]\setminus \{j\}})$. Repeat the process as long as there exist $\sigma_{opt}, i$ and $j$ satisfying the requirements. Note that each application of the step strictly reduces the number of items, so a final instance, denoted by $f_5(C,w,p)$, will eventually be reached. Let $S_5\subseteq S_4$ be the set of instances in $S_4$ such that for every optimal assignment $\sigma_{opt}$ and for every knapsack $i>1$, either $\sigma^{-1}_{opt}(i)=\emptyset$ or $\sigma^{-1}_{opt}(i)=\{n\}$ holds.

\begin{lemma}\label{lem:mk_5}
    $S_4 \to_{f_5} S_5$.
\end{lemma}

\begin{proof}
    It suffices to show the gap-preserving property for one single choice of $\sigma_{opt}, i>1$ and $j\ne n$. By \eqref{eq:mk_lpopt}, it holds that $\lp(C',w',p')=\lp(C,w,p)$, given that $j\ne n$ and both $w_{\text{sum}}$ and $C_{\text{sum}}$ get modified by the same value. On the other hand, by the property of $S_4$, no item can fit into knapsack $i$ with the reduced capacity $C'_i$, except for (potentially) item $n$. In either case, the inclusion-wise maximal feasible assignments of $(C,w,p)$ and of $(C',w',p')$ coincide, with the exception that $j$ can be assigned to an arbitrary knapsack in $(C',w',p')$ due to $w'_j = 0$. It follows that $\ip(C',w',p')=\ip(C,w,p)$ and $\ig(C',w',p')=\ig(C,w,p)$.

    For the second part of the reduction, observe that deleting item $j$ from $(C',w',p')$ decreases both the integer and fractional optimum by $p_j$, so Lemma \ref{lem:gap_change_max} implies $\ig(C'',w'',p'')\ge \ig(C',w',p')$.

    Lastly, let us remark that $f_5$ indeed maps into $S_5$, as it keeps all previous properties of $S_1, \ldots, S_4$.
\end{proof}

In the remaining part, we focus on strengthening the previous property by guaranteeing that there is an optimal assignment $\sigma_{opt}$ that uses only one single item. Let $S_6\subseteq S_5$ be the set of instances in $S_5$ with this property. The reduction will eliminate the few edge cases left by the previous step. For $(C,w,p)\in S_6$, we define $f_6(C,w,p)=(C,w,p)$. Let $(C,w,p) \in S_5 \setminus S_6$, and let $\sigma_{opt}$ be an optimal assignment. We distinguish between the following cases:

\begin{enumerate}
    \item Assume that there exists a knapsack $i>1$ such that $\sigma_{opt}(n)=i$. Any knapsack other than $i$ and $1$ are left empty in $\sigma_{opt}$, and $\sigma^{-1}(1)=j$ for some $j\ne n$. It must hold that $w_{j'} > C_{i'}$ for all items $j'\ne n$ and all knapsacks $i'\ne i$, or else we could strictly increase the profit of $\sigma_{opt}$ by placing $j'\ne j$ on $i$, or we could move $j'=j$ from $1$ to $i$, and place another item on the first knapsack (recall that all items fit in knapsack $1$ by assumption). Observe that item $j$ cannot fit inside knapsack $i$ either, or else the assignment $\sigma'_{opt}$ in which $j$ and $n$ are swapped would be feasible and thus optimal, so the previous reduction step could not have terminated.
    
    \begin{enumerate}
        \item No item, apart from $n$, fits inside knapsack $i$. Delete item $n$ completely, and let $f_6(C,w,p)$ be the resulting instance.
        
        \item There exists an item $k\ne j,n$ such that $w_k \le W_i$. Note that $p_k < p_n$ must hold: if $p_k > p_n$, then swapping them would strictly increase the profit of $\sigma_{opt}$; else if $p_k = p_n$, then swapping them would result in another optimal assignment $\sigma'_{opt}$ for which the previous reduction step is still applicable. Let us define $p'_k:= p_n > p_k$, and $p'=p$ otherwise. For the instance $(C,w,p')$, the assignment $\sigma_{opt}$ can be modified by putting item $k$ on knapsack $i$ and removing $n$, after which the previous reduction step can be repeated. Let $f_6(C,w,p)=f_5(C,w,p')$ be the resulting instance.
    \end{enumerate}
    
    \item In the remaining case, knapsacks $2,\ldots, m$ are empty in $\sigma_{opt}$, but the first knapsack contains more than one item. By property $S_4$, this can only happen if item $n$ and some other item $j\ne n$ are assigned there.

    \begin{enumerate}
        \item There exists an item $k\ne n$ that does not fit into knapsack $1$ together with $n$, because $w_n + w_k > C_1$. Clearly, $k\ne j$ must hold. The idea is to increase $p_k$ and $w_k$ simultaneously in a way that all previous properties remain satisfied, and the gap does not decrease. First, note that $\frac{p_k + \varepsilon}{w_k + \varepsilon}\ge 1=\frac{p_n}{w_n}$ holds for any $\varepsilon>0$, so even though the order of items $1,\ldots, n-1$ might change, item $n$ will remain the last element. By \eqref{eq:mk_lpopt_one}, the modification does not change the LP-optimum, unless item $n$ becomes redundant in $x^*_s$ and stops being the last critical item. This happens when $\sum\limits_{l=1}^{n-1} w_l + \varepsilon = C_{sum}$, so $\varepsilon \le C_{sum} - (w_{sum}-w_n)$ must hold. The other thing we must pay attention to is not to increase the integer optimum. Originally, $p_k \le p_j + p_n$ must hold, or else we could replace the pair $(j,n)$ with $k$ in the first knapsack without decreasing the integer optimum. If the modification does not increase $p_k$ above $p_j + p_n$, the integer optimum will not change either. Let us choose $\varepsilon=\min\{C_{sum} - (w_{sum}-w_n),\, p_j + p_n - p_k\}$, and let us define $w'$ and $p'$ as $p'_k = p_k + \varepsilon, w'_k = w_k + \varepsilon$, and $p'=p, w'=w$ else. If $\varepsilon = C_{sum} - (w_{sum}-w_n)$, we define $f_6(C,w,p)=(C, w'|_{[n-1]}, p'|_{[n-1]})$. In the other case, $\varepsilon = p_j + p_n -p_k$, and we simply define $f_6(C,w,p)=(C,w',p')$.

        \item The pair $(k,n)$ fits inside $C_1$ for all items $k\ne n$. Note that apart from $j$, at least one item must exist, since knapsacks $2,\ldots, m$ cannot be empty in $x^*_s$ by the property $S_3$. Let us define $w'_n = p'_n = \max\{C_2, \ldots, C_m\}$, and $w'=w, p'=p$ else. Since item $n$ can be placed on some machine now, we are in Case $1$. We define $f_6(C,w,p)= f_6 (C, w', p')$ as given above.
    \end{enumerate}

\end{enumerate}

\begin{lemma}\label{lem:mk_6}
    $S_5 \to_{f_6} S_6$.
\end{lemma}

\begin{proof}
    We prove the statement on a case-by-case basis. 
    \begin{itemize}
        \item Case $1.$a. The deletion of item $n$ decreases the LP-optimum by at most $p_n$ and the integer optimum by exactly $p_n$, since the inclusion-wise maximal assignments of $(C,w,p)$ and $(C,w|_{[n-1]}, p|_{[n-1]})$ have a bijection by inserting/moving item $n$ to knapsack $i$.
        
        \item Case $1.$b. Since $k\ne n$, the LP-optimum increases by $p_n-p_k>0$ according to \eqref{eq:mk_lpopt}. We can show that the integer optimum is unchanged. To see this, notice that $w_j > C_i \ge w_n$, or else we could switch items $j$ and $n$ and repeat the previous step. Now if $p_n > p_j$ would be true, we would have that $\nicefrac{p_n}{w_n} > \nicefrac{p_j}{w_j}$, contradicting the minimality of item $n$. Consequently, $p_k \le p_n \le p_j$ holds. Let $\sigma'_{opt}$ be an optimal assignment for $(C,w,p')$. If $\sigma'_{opt}(k)=m+1$, then $p(\sigma'_{opt})=p'(\sigma'_{opt})$, and $\ip(C,w,p)=\ip(C,w,p')$. If $\sigma'_{opt}(k)=1$, then item $n$ is assigned to knapsack $i$ and the total profit is $p'(\sigma'_{opt})=p'_k+p_n = p_n + p_n \le p_j + p_n = \ip(C,w,p)$. Lastly, if $\sigma'_{opt}(k)=i$, then the item with maximal profit (item $j$) is assigned to knapsack $1$, and the total profit is $p'(\sigma'_{opt})=p'_k+p'_j=p_n +p_j = \ip(C,w,p)$. The statement follows.

        \item Case $2$.a. Observe that none of the items fit inside any knapsack $i>1$. If some item other than $j$ does fit, $\sigma_{opt}$ could be strictly improved by placing it in $i$. If item $j$ or $n$ fits in $i$, it could be moved to knapsack $i$ and $f_5$ could be repeated. Lastly, if item $n$ fits in knapsack $i$, we could move it there and would arrive in Case $1$. Let $\sigma'_{opt}$ be an optimal integer assignment of $(C,w',p')$. If $\sigma'_{opt}(k)=1$, then it is the only item there, by the assumption that $w_n + w_k > C_1$. Hence, $\ip(C,w',p') = p'_k = p_k + \varepsilon \le  p_j + p_n =\ip(C,w,p)$. Otherwise, $\sigma'_{opt}(k)=m+1$, and $p'(\sigma'_{opt})=p(\sigma'_{opt})$, implying $\ip(C,w',p') \le \ip(C,w,p)$. As pointed out before, we have $\lp(C,w,p) = \lp(C,w',p')$, so $\ig(C,w',p')=\ig(C,w,p)$. 
        
        In case $\varepsilon = C_{sum} - (w_{sum}-w_n)$ holds, we further have $\lp(C, w'|_{[n-1]}, p'|_{[n-1]}) = \lp(C,w',p')$ and $\ip(C, w'|_{[n-1]}, p'|_{[n-1]}) \le \ip(C,w',p')$, implying $\ig(C,w',p')\le \ig(C, w'|_{[n-1]}, p'|_{[n-1]})$. Since item $n$ is no longer part of the input, the property of $S_6$ indeed holds. In the other case, we can exchange the pair $(j,n)$ with $k$ in knapsack $1$ in $\sigma_{opt}$ to obtain another optimal assignment. It follows that $(C,w',p')\in S_6$ in this case as well.

        \item Case $2$.b. Since $p_n= w_n$ both change by the same value, \eqref{eq:mk_lpopt} remains unchanged. On the other hand, we can show that the integer optimum decreases. Let $\sigma'_{opt}$ be an optimal assignment for $(C,w',p')$. \eqref{eq:mk_lpopt}. Since no two items, apart from $n$, fit together in knapsack $1$, and no items, apart from $n$, fit inside the other knapsacks, we have that the inclusion-wise maximal feasible assignments for $(C,w',p')$ are the ones that contain some item $k\ne n$ in knapsack $1$, and item $n$ in some knapsack (possibly $1$). The profit of these assignments strictly decreased, so $\ip(C,w',p') < \ip(C,w,p)$ holds.
    \end{itemize}
\end{proof}

The property of $S_6$ is strong enough to calculate the gap of an instance explicitly.

\begin{lemma}\label{lem:mk_final}
    $\ig(S_6)=m+1$.
\end{lemma}

\begin{proof}
    Property $S_6$ implies $\ip(C,w,p)=\max_j\{p_j\}$. Since no items fit inside knapsacks $2,\ldots,m$, the instance consist of just one single integer item and at most $m$ fractional items, implying $n\le m+1$. Consequently, if we rewrite \eqref{eq:mk_lpopt} as $\lp(C,w,p)=p_{sum}-\varepsilon$ with $\varepsilon= (w_{sum} - C_{sum})\cdot \nicefrac{p_n}{w_n}$, we have that
    \[
    \ig(C,w,p) = \frac{p_{sum}-\varepsilon}{\max_j\{p_j\}} \le \frac{(m+1)\cdot \max_j\{p_j\}-\varepsilon}{\max_j\{p_j\}}\le m+1.
    \]
    Equality is shown by the following instances: let $p_1 = \ldots = p_n = w_1 \ldots w_n = 1$ with $n=m+1$, and let $C_1 = 2-\frac{1}{n}$ and $C_2 = \ldots = C_{m}=1-\frac{1}{n}$. It is easy to check that the instance belongs to $S_6$, and the gap goes to $m+1$ when $n$ goes to $\infty$.
\end{proof}

\thmmk

\begin{proof}
    It follows from Lemma \ref{lem:mk_final} and from the fact that our reductions is an IGPR-chain.
\end{proof}

\section{Unrelated machine scheduling}\label{app:machine_sched}

Let $(x^*, T^*)$ be an optimal vertex of $\lp(P)$. We say job $j$ is \emph{integer on} $i$ (in $x^*$) if $x^*_{j,i}=1$, and $j$ is \emph{fractional on} $i$ otherwise. Job $j$ is \emph{integer} if it's integer on some machine, and it's \emph{fractional} otherwise. Let $J_{x^*}\subseteq [n]$ denote the jobs that are fractional in $x^*$. We use the same auxiliary bipartite graph $G(x^*)$ as \cite{lenstra_shmoys_tardos} and \cite{Vazirani03} (chapter $17.3$) do for a parametrized version of $\lp(P)$, with results transferable to our case. Let $G(x^*)=(J_{x^*},[m],E)$ where $(j,i)\in E$ if and only if $0 < x^*_{j,i} < 1$. Similar to \cite{Vazirani03} (lemmas $17.3, 17.4$ and $17.6$), we can show $G(x^*)$ is a forest for at least one optimal vertex $(x^*,T^*)$. For the sake of simplicity, we will refer to only $x^*$ as an ``optimal vertex of $\lp(P)$'', since $T^*$ (the fractional optimum) is the same across all optimal vertices. Let $X^*_P$ denote the set of optimal vertices of $P$ such that $G(x^*)$ is a forest.

\begin{lemma}\label{lem:vert_char}
    $X^*_P \ne \emptyset$ for every instance $P$.
\end{lemma}

\begin{proof}
    Let $x^*$ be an optimal vertex of $\lp(P)$, and assume $G(x^*)$ has a cycle. We may assume it is determined by a sequence $j_1, i_1, j_2, i_2, \ldots, j_l, i_l$ such that $x^*_{j_k, i_k}\ne 0$ and $x^*_{j_k, i_{k+1}}\ne 0$ for $k=1, \ldots, l$, with $k+1$ interpreted modulo $l$. Let $\varepsilon = (\varepsilon_1, \ldots, \varepsilon_l)\ge 0$ be an arbitrary vector, and let us define $x^+=x^* +\varepsilon$ as $x^+_{j_1, i_1}:=x^*_{j_1,i_1}+\varepsilon_1, \, x^+_{j_1, i_2}:=x^*_{j_1, i_2}-\varepsilon_1, \ldots, x^+_{j_l, i_l}:=x^*_{j_l,i_l}+\varepsilon_l, \, x^+_{j_l, i_1}:=x^*_{j_l, i_1}-\varepsilon_l$; let $x^+_{j,i}=x^*_{j,i}$ otherwise. Define similarly $x^-= x^*-\varepsilon$ by switching $\pm$ for each $\varepsilon_k$. We want to determine under what conditions $x^+$ or $x^-$ are feasible for $\lp(P)$. Constraint \eqref{job_const} holds for both, by the way we defined the two vectors. In order for \eqref{makespan_const} to hold for $x^+$, it is enough to guarantee that the completion times of machines $i_1, \ldots, i_l$ does not increase when switching from $x^*$ to $x^+$. More precisely, we need, for $k=1, \ldots, l$, that the following condition holds:
    \begin{equation}\label{eq:full_change}
        p_{j_k, i_k} \cdot (x^*_{j_k, i_k}+\varepsilon_k) + p_{j_{k-1}, i_k} \cdot (x^*_{j_{k-1}, i_k}-\varepsilon_{k-1}) \le p_{j_k, i_k} \cdot x^*_{j_k, i_k} + p_{j_{k-1}, i_k} \cdot x^*_{j_{k-1}, i_k},
    \end{equation}
    or equivalently,
    \[
    \varepsilon_k \cdot p_{j_k, i_k} \le \varepsilon_{k-1}\cdot p_{j_{k-1}, i_k}
    \]
    and 
    \begin{equation}\label{eq:eps+}
        \frac{\varepsilon_k}{\varepsilon_{k-1}} \le \frac{p_{j_{k-1}, i_k}}{p_{j_k, i_k}}.
    \end{equation}
    We also need to make sure $x^+ \ge 0$ holds, meaning that $\varepsilon_k \le x^*_{j_k, i_{k-1}}$ needs to be true for all $k$. But once we have chosen $\varepsilon$ such that \eqref{eq:eps+} holds, we can rescale the vector $\varepsilon$ by a constant such that $\varepsilon_k \le x^*_{j_k, i_{k-1}}$ becomes true for all $k$. This rescaling does not affect \eqref{eq:eps+}, as LHS is homogenous in $\varepsilon$.

    For $x^-$, the corresponding inequalities are
    \begin{equation}\label{eq:eps-}
        \frac{p_{j_{k-1}, i_k}}{p_{j_k, i_k}} \le \frac{\varepsilon_k}{\varepsilon_{k-1}}
    \end{equation}
    for $k=1, \ldots, l$.

    We distinguish between two cases. First, let us assume 
    \[
    \frac{p_{j_l, i_1}}{p_{j_1, i_1}} \cdot \frac{p_{j_1, i_2}}{p_{j_2, i_2}} \cdot \ldots \cdot \frac{p_{j_{l-1}, i_l}}{p_{j_l, i_l}} =1.
    \]
    If we define $\varepsilon$ recursively as $\varepsilon_1 := p_{j_1, i_2}$ and $\varepsilon_k := \frac{p_{j_{k-1}, i_k}}{p_{j_k, i_k}}\cdot \varepsilon_{k-1}$ for $k=2, \ldots, l$, then both \eqref{eq:eps+} and \eqref{eq:eps-} hold with equality for all $k$. After an appropriate rescaling of $\varepsilon$, we have $0 \le x^+, x^-$ as well. But then $x^*$ is the convex combination of $x^+$ and $x^-$ (with machine completion times that are coordinate-wise equal to those of $x^*$), contradicting the fact that $x^*$ is a vertex.

    Now assume that
    \[
    \frac{p_{j_l, i_1}}{p_{j_1, i_1}} \cdot \frac{p_{j_1, i_2}}{p_{j_2, i_2}} \cdot \ldots \cdot \frac{p_{j_{l-1}, i_l}}{p_{j_l, i_l}}  > 1.
    \]
    In this case, $x^*+\varepsilon$ with the above-defined $\varepsilon$ is feasible for \eqref{eq:eps+}. Note that if there exists a machine among $i_1, \ldots, i_l$ for which \eqref{makespan_const} is not tight, we can also satisfy \eqref{eq:eps-} as follows. By relabelling, we may assume \eqref{makespan_const} is not tight for $i_l$. Constraints \eqref{eq:eps-} for $i_1, \ldots, i_{l-1}$ are still satisfied (with equality) by the definition of $\varepsilon$, but
    \[
    \frac{p_{j_{l-1}, i_l}}{p_{j_l, i_l}} > \frac{\varepsilon_l}{\varepsilon_{l-1}}
    \]
    might hold. In this case, however, by another rescaling of $\varepsilon$ we can make sure that the LHS of \eqref{eq:full_change} is larger than the RHS only by a small value $\delta$. This does not change the makespan, since constraint \eqref{makespan_const} was not tight for $i_l$ in $x^*$. We reach the same contradiction as before, the vector $x^*$ being the convex combination of $x^+$ and $x^-$.
    
    Consequently, constraint \eqref{makespan_const} is tight for $i_1, \ldots, i_l$. Therefore, the difference vector between $x^*$ and $x^+$ constitutes an edge of the polyhedron $\lp(P)$, as every tight constraint for $x^*$ remains tight for $x^+$, with the exception of \eqref{makespan_const} for machine $i_l$. If we choose $\varepsilon$ such that $x^*+\varepsilon$ is on the boundary of $\lp(P)$, the resulting vector must again be a vertex of $\lp(P)$. This can only happen when $(x^*+\varepsilon)_{j,i}$ becomes $0$ for some coordinate, so $G(x^*+\varepsilon)$ has one less edge than $G(x^*)$, the original cycle is no longer present, and no new cycle was created in the process. If we repeat the argument each time there is a cycle in the auxiliary graph, we eventually arrive at an optimal vertex $x^*$ for which $G(x^*)$ is a forest.
\end{proof}

For the sake of smoother reductions, we focus on instances with at most $M$ machines for a fixed constant $M$, temporarily denoted as $R(\le M)|\cM_j|C_{max}$. A summary of key steps is presented in Table \ref{tab:ms_summary}. The set $S_0 = \{P \in (\enne_+ \cup \{\infty\})^{n \times m}: n,m \in \enne_+, m \le M\}$ denotes the collection of all instances. The first reduction exploits the following observation: for an optimal vertex $x^*$, we can increase the entries of $P$ for which $x^*_{j,i}=0$ without affecting the feasibility or the makespan of $x^*$. The integer optimum can only increase, along with the integrality gap. Let $S_1 := \{ P \in S_0 \vert\, x_{j, i}^* = 0 \Rightarrow p_{j,i} = \infty$ for all $x^*$ optimal vertex of $\lp(P)\}$. Let $P \in S_0\setminus S_1$ such that there exists an optimal vertex $x^*$ and a pair $(j,i)$ for which $x^*_{j,i}=0$ but $p_{j,i}\ne \infty$. Let us fix one such $x^*$, and define $P'$ by setting $p'_{j,i}=\infty$ for all such pairs $(j,i)$, and $p'_{j,i}:=p_{j,i}$ else. It might be that some optimal vertex $x^{**}$ of $\lp(P')$ still does not satisfy the criteria; then we proceed by creating $P''$ based on $x^{**}$, and so on. Eventually, we find an instance $P^f$ whose optimal vertices all satisfy the requirement, since the number of non-infinite coordinates strictly decreases with each modification. Observe that $P^{f}$ might contain a machine $i$ such that $p^f_{j,i}=\infty$ for all jobs $j \in [n]$. We may get rid of such machines, and write $P^f = P^f|_{([m]\setminus \{i\}) \times [n]}$. Let $f_1(P)=P^f$.

\begin{lemma}\label{lem:ms_1}
    $S_0 \to_{f_1} S_1$.
\end{lemma}

\begin{proof}
    It suffices to show the gap-preserving property for a single step of the modification. Let $x^*$ be an optimal vertex of $\lp(P)$, and assume $x^*_{j,i}=0$ but $p_{j,i}\ne \infty$. Let $P'$ be the instance where $p'_{j,i}=\infty$ and $P'=P$ else. The vector $x^*$ remains feasible, so $\lp(P)=\lp(P')$. On the other hand, increasing processing times can only increase the integer optimum, so $\ip(P)\le \ip(P')$. By Lemma \ref{lem:gap_change} Case $1$, we are done.
\end{proof}

Note that for $P \in S_1$, all optimal vertices of $\lp(P)$ have the exact same integer jobs: if $x^*_{j,i}=1$, then $x^*_{j,i'}=0, \,\, \forall i'\ne i$ by \eqref{job_const}. By the property of $S_1$, it must hold that $p_{j,i'}=\infty$ for all $i' \ne i$, so $j$ must be integer on $i$ in every optimal vertex of $\lp(P)$. For this reason, we call a job $j$ \emph{integer (on machine $i$)} if $p_{j,i'}=\infty$ for all $i' \ne i$. Job $j$ is integer iff it is integer in any and all of the optimal vertices of $\lp(P)$. Similarly, if $0< x^*_{j,i}<1$ for some optimal vertex $x^*$, then it must be true for any other vertex as well, and $j$ is called \emph{fractional (on $i$)}. We denote by $J_F \subseteq [n]$ and $J_I \subseteq [n]$ the set of fractional and integer jobs. As a consequence, for $P\in S_1$, every optimal vertex of $\lp(P)$ is automatically in $X^*_P$, and $G(x^*)$ is the same forest (with $V(G(x^*))=J_F \cup [m]$) for each of them, denoted by $G(P)$ from now on.

Our next reduction step puts a limit on the number of integer jobs. Let $S_2 \subseteq S_1$ denote the instances $P \in S_1$ for which there is at most $1$ integer job on each machine. Let $f_2(P)=P$ for $P \in S_2$, else let $f_2(P)=P'$ be defined as follows: if jobs $j_1, \ldots, j_k$ are integer on machine $i$, we merge them into to a new job $j_{new}$ such that $p'_{j_{new}, i} = p_{j_1, i} + \ldots + p_{j_k,i}$, and $p'_{j_{new}, i'}=\infty$ for $i' \ne i$. Repeat this for each machine $i$, and let $P'=P$ elsewhere. 

\begin{lemma}\label{lem:merge_int}
    $S_1 \to_{f_2} S_2$.
\end{lemma}

\begin{proof}
    Optimal vertices of $\lp(P)$ and $\lp(f_2(P))$ are in a one-to-one correspondence with each other. If $x^*$ is an optimal vertex of $\lp(P)$ and $j_1, \ldots, j_k$ are the integer jobs on machine $i$, then define $\hat{x}$ as $\hat{x}_{j_{new},i}=1$, $\hat{x}_{j_{new},i'}=0$ for all $i'\ne i$, and $\hat{x}_{j,i'}=x^*_{j,i'}$ for all $j \not \in \{j_1, \ldots, j_k\}$ and $i' \in [m]$. Repeat this for every machine $i$. The vector $\hat{x}$ is an optimal vertex of $\lp(P')$, and the transformation works in the other direction as well, implying $\lp(P)=\lp(P')$. A similar argument gives that $\ip(P)=\ip(P')$. Note that $f_2(P)$ still satisfies the property of $S_1$, due to the correspondence between the vertices.
\end{proof}

Our next reduction ensures $G(P)$ is a connected graph. Assume it is not the case, and let $G_1 =((J_1, I_1), E_1)$ be a connected component where $J_1 \subseteq J_F$ and $I_1 \subseteq [m]$. Let $J_2 := J_F\setminus J_1$ and $I_2 := [m]\setminus I_1$. Note that $J_1 = \emptyset$ or $J_2 = \emptyset$ might occur, if there are isolated machine nodes in $G$. Let $J_I^1$ and $J_I^2$ denote the integer jobs on machines from $I_1$ and $I_2$, respectively. Let $P_1 = P|_{(J_1\cup J_I^1) \times I_1}$ and $P_2 = P|_{(J_2\cup J_I^2) \times I_2}$. The graph $G_2 = ((J_2, I_2), E_2)$ with $E_2 := E(G(P)) \setminus E_1$ is also a forest, and it holds that $G_1 = G(P_1)$ and $G_2 = G(P_2)$ by the property of $S_1$. Furthermore, the instance $P$ is the block sum of the instances $P_1$ and $P_2$: $\ip(P)=\max\{\ip(P_1), \ip(P_2)\}$ and $\lp(P)=\max\{\lp(P_1), \lp(P_2)\}$, and any optimal vertex $x^*$ of $\lp(P)$ is the block sum of $x^*_1$ and $x^*_1$ that are (not necessarily optimal) vertices of $\lp(P_1)$ and $\lp(P_2)$, respectively.

Assume without loss of generality that $\lp(P_1) \ge \lp(P_2)$, and let $c:= \frac{\lp(P_1)}{\lp(P_2)}$. Let us define $cP_2$ by multiplying each processing time in $P_2$ by $c$. Let $P'$ denote the instance in which $P_2$ is replaced by $cP_2$ in $P$. Clearly, $\ip(cP_2)=c\ip(P_2)$ and $\lp(cP_2)=c\lp(P_2) = \lp(P_1)$. It follows that $\ip(P')\ge \ip(P)$ and $\lp(P')=\lp(P)$. So if we choose $P'' := \arg\max\{\ip(P_1), \ip(cP_2)\}$, it holds that 
\[
\ig(P)=\frac{\ip(P)}{\lp(P)}\le \frac{\ip(P')}{\lp(P')}\le \frac{\ip(P'')}{\lp(P'')}=\ig(P'').
\]
Let $S_3 \subseteq S_2$ denote the instances $P \in S_2$ for which $G(P)$ is connected (and hence it is a tree). Let $f_3: S_2 \to S_3$ be given as $f_3(P)=P$ for $P \in S_3$, else define $P''$ as above. If $G(P'')$ is still not a tree, repeat the same step by considering again a component until $G(P^f)$ is a tree for the final instance $P^f$. Note that $|V(G)|$ decreases strictly in each iteration, so $P^f$ exists. Define $f_3(P):=P^f$. Note that $f_3$ maintains both of the previous properties of $S_1$ and $S_2$, due to the vertices of $\lp(P_2)$ and $\lp(cP_2)$ being the same, and due to the block decomposition of the vertices of $\lp(P)$ into the vertices of the sub-instances. The following Lemma is implied:

\begin{lemma}\label{lem:graph_conn}
    $S_2 \to_{f_3} S_3$.
\end{lemma}

The remaining reductions revolve around shrinking $G(P)$ by deleting special nodes of the tree that have (almost) only leaf neighbours. A leaf of $G(P)$ must be a machine-node, since each fractional job has at least two neighbours in $G(P)$ by definition. We call machine $i$ a \emph{leaf-machine} if $i\in V(G(P))$ is a leaf node. The next preprocessing reduction ensures that leaf-machines don't have integer jobs; let $S_4\subseteq S_3$ denote the set of such instances in $S_3$. Let us define $f_4(P)=P$ for $P \in S_4$. Else, assume $i$ is a leaf-machine and has an integer job $j_{int}$. Since $G(P)$ is a tree and $i$ is a leaf, there must also exist exactly one job $j_{frac}$ that is fractional on $i$. Let us define $P'$ as $p'_{j_{int},i}=0$, $p'_{j_{frac},i}=p_{j_{frac},i} + p_{j_{int},i}$, and let $P'=P$ else. Since $p'_{j_{int,i}}=0$, we can get rid of job $j_{int}$ in the new instance $P'$ without affecting $\lp(P')$ or $\ip(P')$. We repeat the same process for each leaf-machine; let the final instance be $P^f$.

Unlike previous cases, simply defining $f_4(P)=P^f$ would not guarantee that all former properties of $S_1, S_2$ and $S_3$ are satisfied, since vertex structures might change globally with respect to $\lp(P)$. A quick workaround is as follows: if $f_4$ maps $P$ outside $S_1$ (or any of the previous sets), we can repeat the reductions $f_1, f_2$, and $f_3$ and then apply $f_4$ again. If the result is still outside any of the previous sets, we can repeat the whole chain as many times as necessary to arrive at an instance in $S_4$. This will happen in a finite number of iterations, given that each time we use $f_4$, the number of jobs in the instance strictly decreases (with the convention that $j_{int}$ is disposed of after setting $p'_{j_{int}, i}=0$). Let $F_4(P)$ be the result of these finitely many (at most $n$) compositions of the chain $f_3 \circ f_2 \circ f_1 \circ f_4$ with itself. We sketch an auxiliary drawing in Figure \ref{fig:f4_red}.

\begin{figure}
    \centering
    \begin{tikzpicture}
\draw[rounded corners, black] (0, 0) rectangle (6, 7);
\draw[rounded corners, green] (1, 0) rectangle (6, 7);
\draw[rounded corners, red] (2, 0) rectangle (6, 7);
\draw[rounded corners, blue] (3, 0) rectangle (6, 7);
\draw[rounded corners, orange] (4, 0) rectangle (6, 7);

\node at (0.5, 0.5) (s0) {$S_0$};
\node at (1.5, 0.5) (s1) {\textcolor{green}{$S_1$}};
\node at (2.5, 0.5) (s2) {\textcolor{red}{$S_2$}};
\node at (3.5, 0.5) (s3) {\textcolor{blue}{$S_3$}};
\node at (5, 0.5) (s4) {\textcolor{orange}{$S_4$}};

\node at (-1, 7.2) (text) {Number of jobs};

\node at (-1, 6.5) (text) {$2m-1$};
\node at (-1, 6) (text) {$\ldots$};
\node at (-1, 5.5) (text) {$k$};
\node at (-1, 4.5) (text) {$\ldots$};
\node at (-1, 1.5) (text) {$l$};
\node at (-1, 1) (text) {$\ldots$};
\node at (-1, 0.5) (text) {$1$};

\foreach \y in {1,2,3,4,5,6}
    \draw[dashed] (0,\y) -- (6,\y);

\node[draw,circle, inner sep=0pt] (p0) at (3.5,5.5) {$P_0$};
\node[draw,circle, inner sep=0pt] (p1) at (2.5,4.5) {$P_1$};
\node[draw,circle, inner sep=0pt] (p2) at (3.5,4.5) {$P_2$};
\node[draw,circle, inner sep=0pt] (p3) at (0.5,3.5) {$P_3$};
\node[draw,circle, inner sep=0pt] (p4) at (1.5,3.5) {$P_4$};
\node[draw,circle, inner sep=0pt] (p5) at (2.5,3.5) {$P_5$};
\node[draw,circle, inner sep=0pt] (p6) at (3.5,3.5) {$P_6$};
\node[draw,circle, inner sep=0pt] (p7) at (2.5,2.5) {$P_7$};
\node[draw,circle, inner sep=0pt] (p8) at (3.5,2.5) {$P_8$};
\node[draw,circle, inner sep=0pt] (p9) at (5,1.5) {$P_9$};

\draw[->] (p0)  -- node[above] {{\tiny $f_4$}}   (p1);
\draw[->] (p1)  -- node[above] {{\tiny $f_3$}}   (p2);

\draw[->] (p2)  -- node[above] {{\tiny $f_4$}}   (p3);
\draw[->] (p3)  -- node[below] {{\tiny $f_1$}}   (p4);
\draw[->] (p4)  -- node[below] {{\tiny $f_2$}}   (p5);

\draw[->] (p5)  -- node[below] {{\tiny $f_3$}}   (p6);
\draw[->] (p6)  -- node[below] {{\tiny $f_4$}}   (p7);
\draw[->] (p7)  -- node[below] {{\tiny $f_3$}}   (p8);
\draw[->] (p8)  -- node[above] {{\tiny $f_4$}}   (p9);





\end{tikzpicture}
    \caption{Example of the reduction $F_4$. $P_0$ has $k\le 2m-1$ jobs, $P_9 = F_4(P_0)$ has $l \le k-1$ jobs. Reduction steps corresponding to loops are omitted for simplicity.}
    \label{fig:f4_red}
\end{figure}
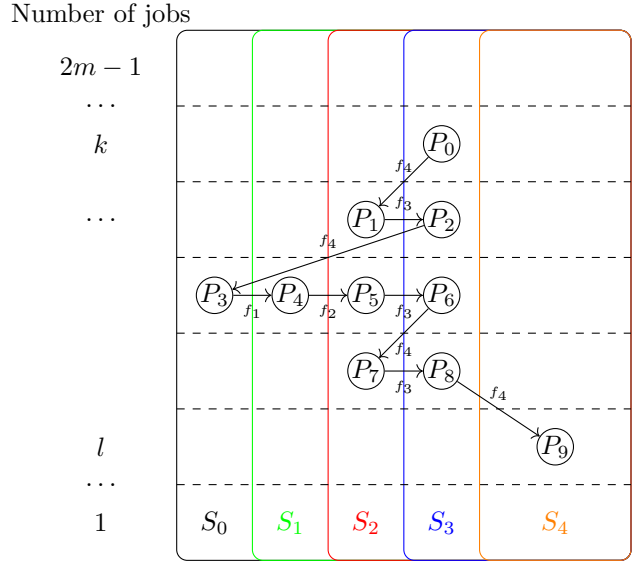

\begin{lemma}\label{lem:int_job}
    $S_3 \to_{F_4} S_4$.
\end{lemma}

\begin{proof}
    It is enough to prove the gap-preserving property of one rebalancing between a pair $j_{frac}$ and $j_{int}$ on a leaf-machine $i$. Let $P\in S_3 \setminus S_3$, let $P'$ be the instance after rebalancing $j_{frac}$ and $j_{int}$. First we prove $\lp(P') \le \lp(P)$. For a simpler analysis, we consider $j_{int}$ being present in $P'$ as well, even though we dispose of it later due to $p'_{j_{int},i}=0$. Let $x^*$ be an optimal vertex of $\lp(P)$. Then $x^*$ is also feasible for $\lp(P')$, and its makespan is at most as large as for $P$: on machine $i$, the completion time changes from $x^*_{j_{frac},i}\cdot p_{j_{frac},i} + x^*_{j_{int},i}\cdot p_{j_{int},i}=x^*_{j_{frac},i}\cdot p_{j_{frac},i} + p_{j_{int},i}$ to $x^*_{j_{frac},i}\cdot p'_{j_{frac},i} + x^*_{j_{int},i}\cdot p'_{j_{int},i}= x^*_{j_{frac},i}\cdot (p_{j_{int},i} + p_{j_{frac},i})$, and the latter is at most the former as $x^*_{j_{frac},i} \le 1$. The completion times of other machines do not change.

    Now we prove $\ip(P)\le \ip(P')$. Consider an optimal assignment $\sigma$ of $P'$. We can assume $\sigma(j_{int})=i$ without loss of generality, given that $p'_{j_{int},i}=0$. We distinguish between two cases. In the first case, $\sigma(j_{frac})=i$. Then $C_{\sigma, P'}(i)=C_{\sigma,P}(i)$ since $p'_{j_{frac},i}+ p'_{j_{int},i}= p_{j_{frac},i} + p_{j_{int},i}$, and the rest of the machine completion times are the same by default. Hence, $\ip(P)\le \ip(P')$. In the second case, $\sigma(j_{frac})=i' \ne i$. Then $C_{\sigma, P'}(i)=0$ and $C_{\sigma, P}(i)=p_{j_{int},i}$; the rest of the completion times are the same in $P$ and in $P'$. Now either $p_{j_{int},i} < C_{max}(\sigma, P')$ and so $\ip(P)\le \ip(P')$, or $p_{j_{int},i} \ge C_{max}(\sigma, P')$. But in the latter case, $p_{j,i}$ must be equal to the makespan of $\sigma$, so $\ip(P)\le C_{max}(\sigma, P)= p_{j_{int},i}\le \lp(P)$ and $\ig(P)=1$. In any case, $\ig(P) \le \ig(P')$ holds.
\end{proof}

The next reduction makes the processing times of fractional jobs on leaf-machines equal to the integer optimum, thus making the instance more uniform. Let $S_5 \subseteq S_4$ be the collection of instances $P \in S_4$ such that any fractional job $j$ on any leaf-machine $i$ s.t. $p_{j,i}\ne \infty$ satisfies $p_{j,i} = \ip(P)$. For $P \in S_5$, let $f_5(P)=P$. Else, if $j$ is fractional on $i$ in $P$, and $p_{j,i} > \ip(P)$, let us define $P'$ as $p'_{j,i}= \ip(P)$ and $P'=P$ else. If instead $p_{j,i} < \ip(P)$ holds, let $p'_{j,i}=0$ and $P'=P$ otherwise. In the latter case, we can dispose of job $j$ entirely without affecting the integer or fractional optima of $P'$, and since $i$ was a leaf-machine, we can delete it as well since $p'_{j',i}=\infty$ holds for all $j' \ne j$. Repeat the above process for each leaf-machine. Let $f_5(P)$ be the final instance, and let $F_5(P)$ be the instance we get by repeatedly composing $F_4 \circ f_3 \circ f_2 \circ f_1 \circ f_5$ with itself until the final instance satisfies all the properties of $S_1, \ldots, S_5$. It happens after finitely many repetations, as each application of $f_5$ either strictly decreases the number of jobs (when $p'_{j,i}=0$), or keeps all the previous properties (when $p'_{j,i}=\ip(P)$).

\begin{lemma}\label{lem:ms_5}
    $S_4 \to _{F_5} S_5$.
\end{lemma}

\begin{proof}
    Let $P \in S_4 \setminus S_5$, and assume $j$ is a fractional job on a leaf-machine $i$, and $p_{j,i} > \ip(P)$. The inequality $\lp(P')\le \lp(P)$ clearly holds, since we only decreased processing times. On the other hand, let $\sigma$ be an arbitrary integer assignment for $P'$. If $\sigma(j)=i$, then $C_{\sigma, P'}(i)=\ip(P)$ and $C_{\sigma, P'}(i')=C_{\sigma, P}(i')$ for all $i'\ne i$. If $\sigma(j)\ne i$, then $C_{\sigma, P'}(i')=C_{\sigma, P}(i')$ for all $i' \in [m]$, since $i$ is a leaf-machine and no job other than $j$ can be assigned there. It follows that $\ip(P')=\ip(P)$, and by by Lemma \ref{lem:gap_change} Case $1$, we are ready.

    Now assume $p_{j,i}< \ip(P)$. The inequality $\lp(P')\le \lp(P)$ clearly holds, as we only reduced processing times. Let $\sigma$ be an integer assignment such that $\sigma(j)=i'\ne i$, and let $\sigma'$ be the assignment in which $\sigma'(j)=i$ and $\sigma'=\sigma$ else. Since $P$ and $P'$ only differ in the $(j,i)$ coordinates, it holds that $C_{max}(\sigma, P')=C_{max}(\sigma, P)$. We also know $C_{max}(\sigma', P)\le C_{max}(\sigma, P)=C_{max}(\sigma, P')$, because moving $j$ from $i'$ to $i$ decreases the completion time of machine $i'$, does not change the completion time apart from machines $i'$ and $i$, and the completion time of $i$ becomes $p_{j,i}< \ip(P)\le C_{max}(\sigma', P)$. Consequently, $C_{max}(\sigma', P')=C_{max}(\sigma',P)$, because setting $p'_{j,i}=0$ does not affect the longest completion time in $\sigma'$. It follows that $\ip(P')=\ip(P)$, and by Lemma \ref{lem:gap_change} Case $1$, we are ready.
\end{proof}

With this, we are ready to start trimming away special nodes of $G(P)$. We call job $j$ a \emph{leaf-job} if it is connected to exactly one non-leaf-machine in $G(P)$. Equivalently, if we remove every leaf-machine from $G(P)$, the leaves of the remaining graph correspond to the leaf-jobs. Note that as long as there are more than $1$ fractional jobs, there always exists at least one leaf-job as well. Consider $G(P)$ as a rooted tree with an arbitrary root $r$, and let $v_d$ be a job-node with maximal depth. Note that $v_d$ is necessarily a leaf-job. Let $u$ be the machine corresponding to the parent of $v_d$. We distinguish between two cases. If there exists a sibling of $v_d$ (a child of $u$ different from $v_d$), it corresponds to another leaf-job $w$. If $v_d$ has no siblings, then $w$ will denote the grandparent of $v_d$ (the parent of $u$). Provided that there are at least two fractional jobs in $P$, one of the two above cases always holds. We sketch an example in Figure \ref{fig:f5_red}.

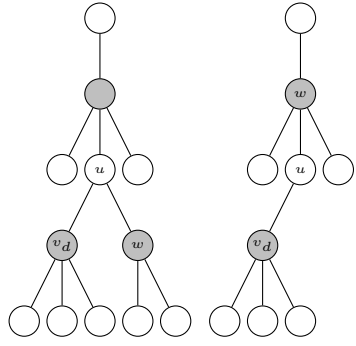
\begin{figure}
    \centering
    \begin{tikzpicture}[scale=1, every node/.style={circle, draw},minimum size=4mm, inner sep=0pt]

 \node (r) at (0,0) {};

\node[fill=gray!50] (r1) at (0,-1) {};

\node (r2) at (-0.5,-2) {};
\node (r3) at (0.5,-2) {};
\node (u) at (0,-2) {{\tiny $u$}};

\node[fill=gray!50] (vd) at (-0.5, -3) {{\tiny $v_d$}};
\node[fill=gray!50] (w) at (0.5, -3) {{\tiny $w$}};

\node (r4) at (-1,-4) {};
\node (r5) at (-0.5,-4) {};
\node (r6) at (0,-4) {};

\node (r7) at (0.5,-4) {};
\node (r8) at (1,-4) {};

\draw (r) -- (r1);
\draw (r1) -- (r2);
\draw (r1) -- (r3);
\draw (r1) -- (u);
\draw (vd) -- (u);
\draw (w) -- (u);
\draw (vd) -- (r4);
\draw (vd) -- (r5);
\draw (vd) -- (r6);

\draw (w) -- (r7);
\draw (w) -- (r8);

\end{tikzpicture}
    \begin{tikzpicture}[scale=1, every node/.style={circle, draw},minimum size=4mm, inner sep=0pt]

 \node (r) at (0,0) {};

\node[fill=gray!50] (w) at (0,-1) {{\tiny $w$}};

\node (r2) at (-0.5,-2) {};
\node (r3) at (0.5,-2) {};
\node (u) at (0,-2) {{\tiny $u$}};

\node[fill=gray!50] (vd) at (-0.5, -3) {{\tiny $v_d$}};

\node (r4) at (-1,-4) {};
\node (r5) at (-0.5,-4) {};
\node (r6) at (0,-4) {};

\draw (r) -- (r1);
\draw (r1) -- (r2);
\draw (r1) -- (r3);
\draw (r1) -- (u);
\draw (vd) -- (u);
\draw (vd) -- (r4);
\draw (vd) -- (r5);
\draw (vd) -- (r6);

\end{tikzpicture}
    \caption{Example of the structure of $G(P)$ rooted at an arbitrary node. $v_d$ is a job node with maximal depth, $w$ is its sibling (if one exists) or its grandparent. Shaded nodes are job-nodes, plain nodes are machine-nodes.}
    \label{fig:f5_red}
\end{figure}

Akin to Lemma \ref{lem:int_job}, we can apply a reduction to get rid of an integer job on machine $u$. Assume job $j$ is integer on $u$. Since $v_d$ is a neighbour of $u$ in $G(P)$, $v_d$ is fractional on $u$. Let $p'_{j,u}:=0, \,\, p'_{v_d, u}:=p_{j,u}+p_{v_d, u}$ and $P'=P$ otherwise. Job $j$ is deleted from $P'$ without changing its integer or fractional optima. Let $S_6 \subseteq S_5$ be the set of instances $P \in S_5$ such that either there is only one fractional job in $P$, or machine $u$ has no integer jobs. Let $f_6(P)=P$ for $P \in S_6$, and $f_6(P)=P'$ for the above-defined $P'$ otherwise. We compose $F_5 \circ F_4 \circ f_3 \circ f_2 \circ f_1 \circ f_6$ with itself as many times as it takes to reach an instance in $S_6$. This will happen in finitely many steps, as each repetition of the chain reduces the number of jobs by at least one. We denote the final instance by $F_6(P)$.

\begin{lemma}\label{lem:ms_6}
    $S_5 \to _{F_6} S_6$.
\end{lemma}

\begin{proof}
    It is enough to prove the gap-preserving property of $f_6$. Let $P\in S_5 \setminus S_6$, and let $P'=f_6(P)$. The inequality $\lp(P')\le \lp(P)$ is proven exactly as in Lemma \ref{lem:int_job}. The other part can be handled similarly as well. Let $\sigma$ be an optimal integer assignment of $P'$. We can assume $\sigma(j)=u$ without loss of generality, given that $p'_{j,u}=0$. We distinguish between two cases. In the first case, $\sigma(v_d)=u$. Then $C_{\sigma, P'}(u)=C_{\sigma, P}(u)$ as $p'_{j,u}+ p'_{v_d, u}=p_{j,u}+p_{v_d, u}$. Furthermore, $C_{\sigma, P'}(i)=C_{\sigma, P}(i)$ holds for all $i\ne u$, as $P$ equals $P'$ aside machine $u$. It follows that $\ip(P)\le \ip(P')$. In the second case, $\sigma(v_d)=i\ne u$. Since $v_d$ is leaf-job, $i$ must be a leaf-machine. We know $\ip(P)=p_{v_d, i}$ from the property of $S_5$, and $p_{v_d, i} = p'_{v_d, i} = C_{\sigma, P'}(i) \le C_{max}(\sigma, P')=\ip(P')$, so $\ip(P) \le \ip(P')$. By Lemma \ref{lem:gap_change} Case $1$, we are done.
\end{proof}

Our final reduction step eliminates one of $v_d$ and $u$, depending on whether $x^*_{v_d,u}$ or $x^*_{w,u}$ is larger in $x^* \in X^*_P$ for $P \in S_6$. In either case, we repeat the same rebalancing of $p_{v_d,u}$ and $p_{w,u}$ from before: if $x^*_{v_d,u} \ge x^*_{w,u}$, then $p'_{v_d,u}:=0, \,\, p'_{w,u}:= p_{v_d,u}+p_{w,u}$ and $P':=P$ else; if instead $x^*_{v_d,u} < x^*_{w,u}$, then $p'_{w,u}:=0, \,\, p'_{v_d,u}:= p_{v_d,u}+p_{w,u}$ and $P':=P$ else. Dispose of the job whose processing time was set to $0$ on $u$. Let $f_7(P)$ be this final instance, and let $f_7(P)=P$ if $P$ has only one fractional job. Repeat the entire chain $F_6 \circ F_5 \circ F_4 \circ f_3 \circ f_2 \circ f_1 \circ f_7$ as many times as it takes to arrive in $S_7= \{P \in S_6: P \text{ has one fractional job}\}$; let $F_7(P)$ denote this final instance.

\begin{lemma}\label{lem:ms_7}
    $S_6 \to_{F_7} S_7$, and $\ig(S_7)=M$.
\end{lemma}

\begin{proof}
    We first prove $\ig(S_7)=M$. Any $P\in S_7$ has just one fractional job $j$. At the same time, $G(P)$ is a tree by the property of $S_3$, so $j$ is fractional on all of the machines, and all machines are leaf-machines. By $S_4$, none of them have integer jobs, $P$ consists of only one job $j$ and some machines $m\le M$. By $S_5$, we know $p_{j,i}=\ip(P)$ for all $i \in [m]$. For such an instance, $\lp(P)=p_{j,1}/m=\ip(P)/m$, so $\ig(P)=m$. The statement follows.
    
    To prove the first statement, it is enough to assert the gap-preserving property of $f_7$. Let $P\in F_6 \setminus F_7$ and $P'=f_7(P)$, and let $x^*$ be an optimal vertex of $\lp(P)$. We distinguish cases according to the relationship between $v_d$ and $w$. First, assume that $v_d$ and $w$ are siblings. Then $w$ is also a leaf-job at maximum depth, so we may assume $x^*_{v_d, u}\ge x^*_{w,u}$ by exchanging $v_d$ and $w$ if necessary. The inequality $\lp(P')\le \lp(P)$ is proven by the fact that $x^*$ has a makespan at most as large for $P'$ as for $P$: on machine $u$, the completion time is $x^*_{v_d, u} \cdot p_{v_d,u} + x^*_{w, u} \cdot p_{w,u}$ for $P$, and $x^*_{w,u}\cdot (p_{v_d,u} + p_{w,u})$ for $P'$. The latter is at most as large as the former, due to the assumption $x^*_{d_v,u}\ge x^*_{w,u}$. The rest of the completion times are the same for $P$ and $P'$. Let $\sigma$ be an optimal assignment for $P'$. We can assume $\sigma(v_d)=u$ without loss of generality, since $p'_{v_d,u}=0$. If $\sigma(w)=u$, then $C_{\sigma, P'}(i)=C_{\sigma, P}$ holds for all $i \in [m]$, implying $\ip(P)\le \ip(P')$. Else, if $\sigma(w)=i\ne u$, then we know $\ip(P')=C_{max}(\sigma, P') \ge C_{\sigma, P'}(i)=p_{w,i} = \ip(P)$, since $w$ is a leaf-job and $i \ne u$ is a leaf-machine, so the property of $S_5$ applies. By Lemma \ref{lem:gap_change}, we are done.

    Note that $v_d$ played no role in the former reasoning, implying that the case when $w$ is the grandparent of $v_d$ and $x^*_{v_d,u} \le x^*_{w,u}$ can be proven the exact same way. The only remaining case is when $w$ is the grandparent of $v_d$ and $x^*_{v_d,u} \ge x^*_{w,u}$. The inequality $\lp(P')\le \lp(P)$ is proven the same way as before. Let $\sigma$ be an optimal integer solution for $P'$. We can assume without loss of generality that $\sigma(v_d)=u$, given that $p'_{v_d,u}=0$. If $\sigma(w)=u$, then $C_{\sigma, P}(i)=C_{\sigma, P'}(i)$ holds for all $i \in [m]$, implying $\ip(P)\le \ip(P')$. Else, $C_{\sigma, P'}(u)=0$ holds by $p'_{v_d, u}=0$, and $C_{\sigma, P'}(i)=0$ for all such $i$ that are leaf-machines and can only accommodate job $v_d$. For all other machines, we have $C_{\sigma, P'}(i)=C_{\sigma, P}(i)$. If $C_{\sigma, P}(u)=p_{v_d,u}> C_{max}(\sigma, P')$, then $C_{max}(\sigma, P)=p_{v_d, u}$ and $\ig(P)=1$. Else, if $C_{\sigma, P}(u)\le C_{max}(\sigma, P')$, then $C_{max}(\sigma, P) \le C_{max}(\sigma, P')$ and $\ip(P)\le \ip(P')$. The desired statement $\ig(P)\le \ig(P')$ follows from Lemma \ref{lem:gap_change}.
\end{proof}

The next theorem follows from the fact that $S_7$ contains instances of $RM||C_{max}$.

\thmms*

\section{The integrality gap of the restricted assignment problem via the configuration LP}\label{app:config}


\subsection{The two-machine case}\label{app:clp2}

Let $S_0$ denote the input universe for $P2|\cM_j|C_{max}$, that is, the union of all $n\times 2$ matrices taken over $n \in \enne$ in which two elements of a row are either the same positive integer, or one of them is $\infty$. Our first reductions eliminate the \emph{machine-exclusive} jobs: we say that job $j$ is exclusive relative to machine $i$ if $j$ can only be assigned to $i$. Let $S_1 \subset S_0$ be composed of all instances where both machines have at most one exclusive job. In other words: $S_1$ is the set of matrices of $S_0$ in which each column has at most one $\infty$-entry. We define $f_1: S_0 \to S_1$ as $f_1(P)=P$ for all $P \in S_1$. Else, if column $2$ of $P$ has more than one $\infty$-entry, we can assume by relabelling them that these belong to jobs $j_1, \ldots j_k$. Let us define the instance $P'$ by merging jobs $j_1, \ldots, j_k$ together to a new job $j_{new}$, where $p_{j_{new}} = p_{j_1} + \ldots +p_{j_k}$ and $j_{new}$ is $1$-exclusive as well. Apply the same modification to column $1$ if necessary, and let $f_1(P)=P'$.

\begin{lemma}\label{lem:exclusive}
    $S_0 \to_{f_1} S_1$.
\end{lemma}

\begin{proof}
    Let $P \in S_0 \setminus S_1$, and let $x$ be a feasible solution for $\clp(P,T)$ for some $T$. Let jobs $j_1,\ldots, j_k$ be all the $1$-exclusive jobs in $P$. These can only appear in $x$ in configurations of machine $1$, hence \eqref{machine_const} for machine $1$ and \eqref{ass_const} for $j_1,\ldots, j_k$ together imply that $\{j_1, \ldots, j_k\} \subseteq C$ for all $C \in \cC_1(P,T)$ such that $x_{1,C}>0$. Let us define $\hat{x}$ as follows: for $C \in \cC_1(P,T)$ where $x_{1,C}>0$, let $C' = C \cup \{j_{new}\} \setminus \{j_1, \ldots, j_k\}$ and let $\hat{x}_{1,C'} = x_{1,C}$. Let $\hat{x}_{2,C}=x_{2,C}$ for every $C \in \cC_2(P,T)$ and let $\hat{x}_{i,C}=0$ for any other $i$ and $C$. The vector $\hat{x}$ is feasible for $\clp(P', T)$ for $P' = f_1(P)$, since we only replaced $\{j_1, \ldots, j_k\}$ with $j_{new}$ in all configurations that contain them, so \eqref{machine_const} and \eqref{ass_const} are implied immediately for all machines and jobs in $P'$ except for $j_{new}$. For the latter, \eqref{ass_const} is implied by the fact that all configurations of $x$ on machine $1$ contain each of $j_1, \ldots, j_k$. The inequality $\lp(P') \le \lp(P)$ follows.

    On the other hand, the integer feasible solutions of $P$ and $P'$ are in a direct one-to-one correspondence which maintains the machine completion times as well: for any feasible integer solution of $P$, replace the collection of jobs $\{1,\ldots, k\}$ (necessarily assigned to machine $1$) with the new job $j_{new}$, and vice versa. By the way we defined $p_{j_{new}}$ this operation does not change the makespan of the two solutions. Hence, $\ip(P') = \ip(P)$. By Lemma \ref{lem:gap_change} (Case 2), we are done.
\end{proof}

At this point, both machines have at most $1$ exclusive job. As a next step, we can get rid of one of them in case both machines have exactly one. Let $S_2 \subseteq S_1$ be the instances in which there is at most one machine-exclusive job for machines $1$ and $2$ combined. Let us define $f_2: S_1 \to S_2$ as $f_2(P) = P$ if $P \in S_2$; else, assume that $j_1$ is $1$-exclusive and $j_2$ is $2$-exclusive. Assume without loss of generality that $p_{j_1} \le p_{j_2}$. Let us define $P'$ the following way: for any $j\ne j_1,j_2$, the processing times are the same as in $P$. Let $p'_{j_2} = p_{j_2} - p_{j_1}$, $p'_{j_1} = p_{j_1} - p_{j_1} = 0$, and let $j_1$, $j_2$ remain $1$-exclusive and $2$-exclusive, respectively. Since $p'_{j_1} = 0$, we can further remove $j_1$ from the set of jobs in $P'$. Let $f_2(P) = P'|_{([n]\setminus \{j_1\}) \times \{1,2\}}$.

\begin{lemma}\label{lem:clp2_2}
    $S_1 \to_{f_2} S_2$.
\end{lemma}

\begin{proof}
    Let $P \in S_1 \setminus S_2$, and let $j_1$ and $j_2$ be $1$-exclusive and $2$-exclusive in $P$, respectively. Let $x$ be any feasible solution of $\clp(P,T)$. Akin to the previous proof, for each $C \in \cC_1(P,T)$ such that $x_{1,C}>0$, $j_1 \in C$ must hold. Likewise, $j_2 \in C$ for all $C \in \cC_2(P,T)$ where $x_{2,C}>0$. Consequently, if we decrease the processing times of $j_1$ and $j_2$ simultaneously by $p_{j_1}$, the combined processing time of each configuration in $x$ decreases by $p_{j_1}$ as well. Hence, $x$ is feasible for $\clp(P', T-p_{j_1})$, and it can be easily checked that the reverse implication is true as well, implying $\lp(P')=\lp(P)-p_{j_1}$.

    At the same time, the simultaneous decrease of $p_{j_1}$ and $p_{j_2}$ affects any integer solution of $P$ the same way: since $j_1$ must be assigned to machine $1$ and $j_2$ to machine $2$, both machine completion times reduce by $p_{j_1}$ as well. When increasing back the processing times from $p'_{j_1}$ and $p'_{j_2}$ by $p_{j_1}$, any feasible integer solution of $P'$ increases by $p_{j_1}$ on both machines. Hence, $\ip(P') = \ip(P)-p_{j_1}$ as well. The statement follows from Lemma \ref{lem:gap_change} (Case $3$).
\end{proof}

Now each instance has at most one exclusive job. In the next step, we can dispose of exclusivity. Let $S_3 \subseteq S_2$ be the set of all instances without machine-exclusive jobs. Let us define $f_3(P) = P$ for $P \in S_3$. Otherwise, assume $p_{j_1,1}=\infty$ and let $f_3(P)=P'$ where $P'$ is identical to $P$, with the sole exception that $p'_{j_1,1}=p'_{j_1,2} = p_{j_1}$.

\begin{lemma}\label{lem:clp2_3}
    $S_2 \to_{f_3} S_3$.
\end{lemma}

\begin{proof}
    Let $P \in S_2 \setminus S_3$, and assume $p_{j_1,1} = \infty$. Let $\hat{x}$ be a feasible solution for $\clp(P',T)$ for $P' = f_3(P)$ and some $T$. Notice that since $P'$ does not have machine-exclusive jobs, we may write $\cC_1(P',T) = \cC_2(P', T) := \cC(P', T)$. Let us define a vector $x$ as follows: $x_{2,C} := \hat{x}_{1,C} + \hat{x}_{2,C}$ for all $C \in \cC(P', T)$ with $j_1 \in C$, and $x_{1,C} := \hat{x}_{1,C} + \hat{x}_{2,C}$ for all $C \in \cC(P', T)$ such that $j_1 \notin C$; and let $x$ be zero elsewhere. The vector $x$ satisfies machine constraints \eqref{machine_const}, because on machine $2$ it contains all variables of $\hat{x}$ that contain $j_1$ (which sum up to $1$ by \eqref{ass_const}), and on machine $1$ it contains the rest of $\hat{x}$. The latter sum equals $1$ as well, because it equals the total sum of $\hat{x}$ ($= 2\cdot 1$ by \eqref{machine_const}) minus the sum of $\hat{x}$ where it contains $j_1$. Furthermore, $x$ also satisfies the job constraints \eqref{ass_const}, because we just redistributed variables of $\hat{x}$ on the two machines. On the other hand, any feasible solution for $\lp(P,T)$ is automatically feasible for $\lp(P',T)$. Consequently, $\lp(P') = \lp(P)$.

    Similarly, any integer assignment for $P'$ can be transformed easily into another one with the same makespan that does not contain $j_1$ on machine $1$: if it does, just simply switch the label of the two machines. It follows that $\ip(P')=\ip(P)$, and $\ig(P')=\ig(P)$.
\end{proof}

The set $S_3$ corresponds to instances of the identical machine scheduling on two machines, denoted by $P2||C_{max}$. This knowledge will be enough to determine the gap exactly, based on two simple observations. The first one claims that for any feasible solution $x$ of $\clp(P, \lp(P))$, there exists a configuration $C$ on some one of the machines (say $i$) that has a completion time $\lp(P)$ and $x_{i,C}>0$. Indeed, if not, the same $x$ would be feasible for $\clp(P, \lp(P)-1)$, contradicting the minimality of $\lp(P)$.

The second observation asserts that for an arbitrary instance $P$ for the identical machine scheduling problem with $m$ machines, the inequality $\frac{p_1 + \ldots + p_n}{m} \le \lp(P)$ holds. Indeed, let us examine the expression $\sum\limits_{i,C} x_{i,C}\cdot p(C)$ for a feasible solution $x$, where $p(C)$ denotes the total processing time of a configuration $C$. On one hand, each job's processing time appears with a total coefficient of $1$ due to \eqref{ass_const}, so the sum is $p_1 + \ldots + p_n$. On the other hand, for any machine, the sum of corresponding coefficients is $1$ due to \eqref{machine_const}, and each $p(C)$ is at most $\lp(P)$ due to the feasibility of $x$. Thus, the sum is at most $m \cdot \lp(P)$, and the statement follows by rearrangement.

\begin{lemma}\label{lem:iden}
    $\ig(P2||C_{max}) = 1$.
\end{lemma}

\begin{proof}
    Let $P$ be an arbitrary instance of $P2||C_{max}$, and let $x$ be a feasible solution for $\clp(P,\lp(P))$. By our first observation, we can assume that there exists $C \in \cC_1(P, \lp(P))$ such that $p(C) = \lp(P)$ and $x_{1,C} > 0$. Let $\overline{C} = [n]\setminus C$ denote the rest of the jobs. From the second observation, it holds that $p(\overline{C}) = p([n]) - p(C)\le 2\cdot \lp(P) - \lp(P) = \lp(P)$. Hence, if we put all jobs in $C$ to machine $1$ and all jobs in $\overline{C}$ to machine $2$, we get an integer assignment with a makespan of $\lp(P)$. It follows that $\ig(P)=1$.
\end{proof}

\thmmequaltwo*

\begin{proof}
    It follows from Lemma \ref{lem:iden} and from the fact that $((S_0, S_1, S_2, S_3), (f_1, f_2, f_3))$ is an IGPR-chain.
\end{proof}

\begin{remark*}
    Let us point out that Theorem \ref{thm:p2} does not contradict the $\mathsf{NP}$-hardness of $P2|\mathcal{M}_j|C_{max}$. First, a gap of $1$ does not imply that the polyhedron of $\clp(P,\lp(P))$ is integer; several fractional vertices do exist for many instances. Hence, solving the LP-relaxation to optimality does not automatically return an integer feasible solution. Second, the time complexity of solving $\clp(P,\lp(P))$ is exponential in $(P,n,m)$, since it may contain $O(2^n)$-many constraints.
\end{remark*}

\subsection{Unit processing times}\label{app:unit}

Let us fix an instance $P$, and consider an integer optimal assignment $\sigma$ with makespan $\ip(P)$. We construct an auxiliary digraph $D(\sigma) = (V,A)$ that encodes the possible local transformations of $\sigma$ to another feasible (not necessarily optimal) integer assignment. For each machine $i \in [m]$, we include a node $v_i \in V$. We add an arc $\overrightarrow{v_{i_1} v_{i_2}}$ to $A$ if there exists a job $j$ such that $\sigma(j)=i_1$, and it is allowed on machine $i_2$ as well (that is, $p_{j, i_2} \ne \infty$). For simplicity, we only include one arc for the entire set of jobs that could be moved from $i_1$ to $i_2$. Let $i_{max}$ be a machine which obtains the makespan $\ip(P)$ in $\sigma$. The key observation will be that when $\sigma$ is chosen in a clever way, we can assume every other node is reachable from $v_{i_{max}}$ in $D(\sigma)$ via a directed path. In particular, let $I_{max}(\sigma)=\{i \in[m]: C_{\sigma}(i)=\ip(P)\}$ denote the set of all machines whose completion time equals $\ip(P)$ according to $\sigma$. We will consider $\sigma$ for which $|I_{max}(\sigma)|$ is minimal among all integer optimal assignments; let us denote it by $\sigma_f$ indicating a minimal number of full machines.

With this, we are ready to construct the reduction. Let $S_0$ denote all the instances of $P|\mathcal{M}_j, p_j = 1|C_{max}$. Let $S_1 \subseteq S_0$ be the subset of instances for which $\sigma_f$ satisfies that for every $i_{max}\in I_{max}(\sigma_f)$, every node is reachable from $v_{i_{max}}$ in $D(\sigma_f)$. Define $f_1(P)=P$ for $P \in S_1$. Otherwise, assume $i_{max}\in I_{max}(\sigma_f)$ is such that not all nodes are reachable from $v_{i_{max}}$ in $D(\sigma_f)$. Let $I_1$ be the set of machines $i\in[m]$ for which $v_i$ is reachable from $v_{i_{max}}$ (including $v_{i_{max}}$ itself), and let $I_2 = [m]\setminus I_1$. Let $J_1$ be the set of jobs assigned to machines from $I_1$ by $\sigma_f$, and let $J_2 = [n]\setminus J_1$. We define $f_1(P) = P|_{J_1 \times I_1}$ as the restriction of $P$ to machines $I_1$ and jobs $J_1$.

\begin{lemma}\label{lem:graph_connect}
    $S_0 \to_{f_1} S_1$.
\end{lemma}

\begin{proof}
    Let $P\in S_0 \setminus S_1$, and let $i_{max} \in I_{max}(\sigma_f)$ be a machine such that not all nodes are reachable from $v_{i_{max}}$ in $D(\sigma_f)$. The crucial observation is that $p_{j,i}=\infty$ for all $j\in J_1$ and $i \in I_2$. Indeed, $p_{j,i}\ne\infty$ for such $j$ and $i$ would imply that there is an arc from $v_{\sigma_f(j)}$ to $v_i$. But $j\in J_1$ means $v_{\sigma_f(j)}$ is reachable from $v_{i_{max}}$, so $v_i$ would be reachable as well, contradicting $i \in I_2$. Note that a similar statement is not necessarily true for $J_2$ and $I_1$; jobs from $J_2$ could very well be allowed on some machine in $I_1$.

    As a consequence, $\lp(P')\le \lp(P)$ holds for $P'=f_1(P)$. To see it, let $x$ be an arbitrary feasible solution of $\clp(P,T)$ for some $T$. By the above observation, $x|_{J_1 \times  I_1}$ satisfies \eqref{ass_const} for each job in $J_1$. Constraints \eqref{machine_const} remain satisfied for machines in $I_1$. Thus, the vector $\hat{x}_{i,C\setminus J_2} = x_{i,C}$ for all $i \in I_1$ and $C \in \cC_i (P,T)$ is feasible for $\lp(P',T)$.

    At the same time, $\ip(P')=\ip(P)$ holds. If not, the jobs in $J_1$ could be rearranged in $\sigma_f$ such that the highest completion time in $I_1$ would be at most $\ip(P')\le \ip(P)-1$. But this would mean that in the modified assignment, no machine in $I_1$ (including $i_{max}$) has a completion time equal to $\ip(P)$, contradicting the minimality of $\sigma_f$. The statement follows from Lemma \ref{lem:gap_change} Case $2$.
\end{proof}

\begin{lemma}\label{lem:graph_moving}
    $\ig(S_1)=1$.
\end{lemma}

\begin{proof}
    Let $P\in S_1$, and consider the optimal assignment $\sigma_f$. Let $i_{max}\in I_{max}(\sigma_f)$. Note that there cannot exists a machine $i$ for which $C_{\sigma_f}(i) \le \ip(P)-2$. Indeed, since $v_i$ is reachable from $v_{i_{max}}$ in $D(\sigma_f)$, we could consider a shortest directed path $v_{i_{max}}=v_{i_0}, v_{i_1}, \ldots, v_{i_k} = v_{i}$ between them and move some job $j_l \in \sigma_f^{-1}(i_l)$ from $i_l$ to $i_{l+1}, \,\,l = 0, \ldots, k-1$. As the path is the shortest, the moved jobs have no repetitions, and so the completion times of machines $i_1, \ldots, i_{k-1}$ do not change. The completion time $C_{\sigma_f}(i_{max})$ decreases to $\ip(P)-1$, and $C_{\sigma_f}(i)$ increases to at most $\ip(P)-1$, contradicting the minimality of $\sigma_f$.

    Thus, each completion time is either $\ip(P)$, or $\ip(P)-1$, which implies that $C_{max}(\sigma_f)$ is equal to the trivial lower bound $\left\lceil\nicefrac{\sum_{j} p_j}{m}\right\rceil=\lceil\frac{n}{m}\rceil$ on $\lp(P)$. It follows that $\ig(P)=1$.
\end{proof}

Lemmas \ref{lem:graph_connect} and \ref{lem:graph_moving} imply the following theorem:

\thmunit*

\subsection{Improving the gap of known instances}

In this section, we put our method into practice to find instances of $P|\cM_j|C_{max}$ with a higher integrality gap than those in \cite{Kurpisz18}. Our reduction framework allows us to slightly strengthen this bound, increasing it from $1.00098$ to $1.01373$. The construction relies on a sequence of computer-aided reductions. We begin with the instance from \cite{Kurpisz18}, consisting of $n=15$ jobs and $m=3$ machines, with input matrix

\[
P_1 =
\left[
\begin{array}{*{15}{c}}
3 & 6 & 12 & 24 & 17 & 33 & 66 & 132 & 264 & 528 & 160 & 640 & 576 & 320 & 288 \\
3 & 6 & 12 & 24 & 17 & 33 & 66 & 132 & 264 & 528 & 160 & 640 & 576 & 320 & 288 \\
3 & 6 & 12 & 24 & 17 & 33 & 66 & 132 & 264 & 528 & 160 & 640 & 576 & 320 & 288
\end{array}
\right]
\]

The corresponding optimal LP solution proposed in \cite{Kurpisz18} is not half-integral. However, we construct an equivalent half-integral solution, in which each of the following (machine, configuration) pairs is assigned the value $\nicefrac{1}{2}$:


\begin{center}
\begin{tabular}{c|c}
    \hline
    Machine & Configuration\\
    \hline
     \multirow{2}{*}{1} & \{1 3 10 11 14\}  \\
                        & \{1 4 8 13 15\}  \\ \hline
     \multirow{2}{*}{2} & \{2 5 9 11 13\}  \\
                        & \{6 7 8 9 10\}  \\ \hline
     \multirow{2}{*}{3} & \{2 4 6 12 14\} \\
                        & \{3 5 7 12 15\} \\
     \hline
\end{tabular}    
\end{center}

Observing that each job $j$ appears in at most two machine configurations, we can restrict the instance by retaining only the machines which contain a configuration that includes the given job, and is assigned a non-zero value. This yields the following instance:
\[
P_1 =
\left[
\begin{array}{*{15}{c}}
3 & \infty & 12 & 24 & \infty & \infty & \infty & 132 & \infty & 528 & 160 & \infty & 576 & 320 & 288 \\
\infty & 6 & 12 & \infty & 17 & \infty & 66 & \infty & 264 & \infty & 160 & 640 & 576 & \infty & 288 \\
\infty & 6 & \infty & 24 & \infty & 33 & 66 & 132 & 264 & 528 & \infty & 640 & \infty & 320 & \infty
\end{array}
\right]
\]

Similar to Lemma \ref{lem:ms_1}, one can show the above modification preserves the gap. In particular, it increases to $\nicefrac{1037}{1024}$. We now apply a second transformation based on a simple and more general observation. Let $x$ be a half integral optimal solution, and let $S(x)$ be a set of jobs that never appear together in any configuration. In our example, we chose $S(x)$ to be $\{0,4,5\}$. Define $q(x) := \min_{j \in S(x)} p_j$, and consider the modified instance given by
\[
p'_j := \begin{cases}
p_j & j \notin S(x), \\
p_j - q(x) & j \in S(x).
\end{cases}
\]
By construction, at least one job in $S(x)$ has processing time reduced to zero; we remove this job, obtaining an instance with $n' = n - 1$ jobs. We claim that $\ip(I') = \ip(I) - q(x)$ and $\lp(I') = \lp(I) - q(x)$. To see this, consider an optimal integral (respectively, fractional) solution for $I'$. If its value was strictly smaller than $\ip(I) - q(x)$ (respectively, $\lp(I) - q(x)$), we could add back $q(x)$ units of processing time to each machine by reintroducing the contribution of the jobs in $S(x)$, increasing the makespan by at most $q(x)$. This would yield a solution for the original instance with value strictly smaller than $\ip(I)$ (respectively, $\lp(I)$), contradicting optimality. From Lemma \ref{lem:gap_change} Case $2$, it follows that the gap increases. Indeed, the new instance has a gap of $\nicefrac{1034}{1020} \sim 1.01373$ and it has the following structure:

\[
P_2 =
\left[
\begin{array}{*{14}{c}}
\infty & 12 & 24 & \infty & \infty & \infty & 132 & \infty & 528 & 160 & \infty & 576 & 320 & 288 \\
 6 & 12 & \infty & 14 & \infty & 66 & \infty & 264 & \infty & 160 & 640 & 576 & \infty & 288 \\
 6 & \infty & 24 & \infty & 30 & 66 & 132 & 264 & 528 & \infty & 640 & \infty & 320 & \infty
\end{array}
\right]
\]

\section{Connection to the bin-packing problem}

In the bin packing problem, an unlimited amount of fix-sized bins are provided, along with a collection of items with specific weights that are need to be packed into as few bins as possible, such that total weight of items hosted in a specific bin does not exceed the bin capacity.

The parallelism with the scheduling framework is evident: in scheduling problems, the number of machines/bins is fixed and we have to minimize the containers' capacity/completion time, whereas the bin packing framework inverts their roles, and aims to minimize the number of containers while keeping the capacity fixed. We are going to translate Theorem \ref{thm:p2} to a result for bin packing.

\end{document}